\documentclass[10pt, onecolumn]{IEEEtran}
\usepackage[utf8]{inputenc}
\usepackage{amsmath}
\usepackage{mathtools}
\usepackage{amsfonts}
\usepackage{amssymb}
\usepackage{amsthm}
\usepackage{subfigure}
\usepackage{cite}
\usepackage{calc}
\usepackage{color}
\usepackage{epsfig}
\usepackage{setspace}
\usepackage{pstricks}
\usepackage{cancel}
\usepackage{multirow}

\newtheorem{defn}{Definition}
\newtheorem{thm}{{\cal T}heorem}[section]
\newtheorem{cor}[thm]{Corollary}
\newtheorem{prop}{Proposition}
\newtheorem{lem}[thm]{Lemma}
\newtheorem{conj}[thm]{Conjecture}
\newtheorem{constr}[thm]{Construction}
\newtheorem{note}{Remark}

\newtheorem{example}{Example}
\newcommand{\bit}{\begin{itemize}}
\newcommand{\eit}{\end{itemize}}
\newcommand{\bcor}{\begin{cor}}
\newcommand{\ecor}{\end{cor}}
\newcommand{\beq}{\begin{equation}}
\newcommand{\eeq}{\end{equation}}
\newcommand{\beqn}{\begin{equation*}}
\newcommand{\eeqn}{\end{equation*}}
\newcommand{\bea}{\begin{eqnarray}}
\newcommand{\eea}{\end{eqnarray}}
\newcommand{\bean}{\begin{eqnarray*}}
\newcommand{\eean}{\end{eqnarray*}}
\newcommand{\ben}{\begin{enumerate}}
\newcommand{\een}{\end{enumerate}}
\newcommand{\bdefn}{\begin{defn}}
\newcommand{\edefn}{\end{defn}}
\newcommand{\bnote}{\begin{note}}
\newcommand{\enote}{\end{note}}
\newcommand{\bprop}{\begin{prop}}
\newcommand{\eprop}{\end{prop}}
\newcommand{\blem}{\begin{lem}}
\newcommand{\elem}{\end{lem}}
\newcommand{\bthm}{\begin{thm}}
\newcommand{\ethm}{\end{thm}}
\newcommand{\bconj}{\begin{conj}}
\newcommand{\econj}{\end{conj}}
\newcommand{\bconstr}{\begin{constr}}
\newcommand{\econstr}{\end{constr}}
\newcommand{\bpf}{\begin{proof}}
\newcommand{\epf}{\end{proof}}
\newcommand{\bprf}{{\em Proof: }}
\newcommand{\eprf}{\hfill $\Box$}

\DeclarePairedDelimiter\ceil{\lceil}{\rceil}
\DeclarePairedDelimiter\floor{\lfloor}{\rfloor}

\newcommand{\cauchy}[3]{\g_{1,#1_1} & \ldots & \g_{1,#1_{#2}} \\
	\vdots & \ddots & \vdots \\
	\g_{#3,#1_1} & \ldots & \g_{#3, #1_{#2}}}

\newcommand{\cauchyrowm}[4]{#4_{#1} \cdot \g_{#3,#1_1} & \ldots & #4_{#1} \cdot \g_{#3,#1_{#2}}}
\newcommand{\cauchyrowp}[3]{\g_{#3,#1_1} & \ldots & \g_{#3,#1_{#2}}}

\newcommand{\Fq}[1]{
	\mathbb{F}_{q^{#1}}
}
\newcommand\scalemath[2]{\scalebox{#1}{\mbox{\ensuremath{\displaystyle #2}}}}

\newcommand{\Fqo}{\mathbb{F}_{q_0}}
\newcommand{\A}{\alpha}
\newcommand{\B}{\beta}

\newcommand{\dl}{\delta}
\newcommand{\g}{\gamma}
\newcommand{\La}{\lambda}

\newcommand{\eal}{\end{align*}}
\newcommand{\bbm}{\begin{bmatrix}}
\newcommand{\ebm}{\end{bmatrix}}
\newcommand{\bpm}{\begin{pmatrix}}
\newcommand{\epm}{\end{pmatrix}}
\newcommand{\bal}{\begin{align*}}
    \newcommand{\bfg}{\begin{figure*}}
	\newcommand{\efg}{\end{figure*}}
	
\newcommand{\mattwoone}[2]{
\left(
  \begin{matrix}
  #1\\
  #2\\
  \end{matrix}
  \right)}

  \newcommand{\matthreeone}[3]{
\left(
  \begin{matrix}
  #1\\
  #2\\
  #3\\
  \end{matrix}
  \right)}

\title{Maximally Recoverable Codes with Hierarchical Locality: Constructions and Field-Size Bounds}

\author{D. Shivakrishna, Aaditya M. Nair and V. Lalitha
\thanks{D. Shivakrishna, Aaditya M. Nair and V. Lalitha are with Signal Processing and Communications Research Centre, International Institute of Information Technology, Hyderabad,
 India (email: d.shivakrishna@research.iiit.ac.in, aadityamnair@research.iiit.ac.in, lalitha.v@iiit.ac.in).}
\thanks{The work of V. Lalitha is supported partly by the grants ECR/2016/000954 and MTR/2017/000806 from Science and Engineering Research Board (SERB).} 
\thanks{This paper was presented in part at the 2019 National Conference on Communications \cite{nair2019maximally} and to be presented in part at 2021 International Symposium on Information Theory (ISIT).}}

\begin{document}

\maketitle

 \begin{abstract}
 Maximally recoverable codes are a class of codes which recover from all potentially recoverable erasure patterns given the locality constraints of the code. In earlier works, these codes have been studied in the context of codes with locality. The notion of locality has been extended to hierarchical locality, which allows for locality to gradually increase in levels with the increase in the number of erasures. We consider the locality constraints imposed by codes with two-level hierarchical locality and define maximally recoverable codes with data-local and local hierarchical locality. We derive certain properties related to their punctured codes and minimum distance. We give a procedure to construct hierarchical data-local MRCs from hierarchical local MRCs. We provide a construction of hierarchical local MRCs for all parameters. We also give constructions of MRC with hierarchical locality for some parameters, whose field size is smaller than that of known constructions for general parameters.  We also derive a field size lower bound on MRC with hierarchical locality.

 \end{abstract}

\section{Introduction}

With application to distributed storage systems, the notion of locality of a code was introduced in \cite{gopalan2012locality}, which enables efficient node repair in case of single node failures (node failures modelled as erasures)  by contacting fewer nodes than the conventional erasure codes based on maximum distance separable (MDS) codes. An extension to handle multiple erasures has been studied in \cite{locality}.
A code symbol is said to have $(r, \delta)$ locality if there exists a punctured code $\mathcal{C}_i$ such that $c_i \in Supp(\mathcal{C}_i)$ and the following conditions hold, 
\begin{itemize}
    \item $dim(\mathcal{C}_i) \leq r$ and
    \item $d_{min}(\mathcal{C}_i) \geq \delta$
\end{itemize}
An $[n,k,d_{min}]$ code is said to have $(r,\delta)$ information locality, if $k$ data symbols have $(r,\delta)$ locality and it is said to have all-symbol locality if all the $n$ code symbols have $(r,\delta)$ locality. An upper bound on the minimum distance of a code with $(r,\delta)$ information locality is given by
\begin{equation} \label{eq:dmin_rdelta}
d_{min} \leq n - k + 1 - \left ( \left \lceil \frac{k}{r} \right \rceil -1 \right ) (\delta-1).
\end{equation}

\subsection{Maximally Recoverable Codes with Locality}

Maximally recoverable codes (MRC) are a class of codes which recover from all information theoretically recoverable erasure patterns given the locality constraints of the code. Maximally recoverable codes with locality have been defined for the case of $\delta =2$ in \cite{mrc}. We extend the definitions here for the general $\delta$.

\begin{defn}[Data Local Maximally Recoverable Code]\label{defn:data_local}
Let $C$ be a systematic $[n,k,d_{min}]$ code. We say that $\mathcal{C}$ is an $[k, r, h, \delta]$ data-local maximally recoverable code if the following conditions are satisfied
\bit
    \item $r | k$ and $n=k+\frac{k}{r}\cdot \delta+h$
    \item Data symbols are partitioned into $\frac{k}{r}$ groups of size $r$. For each such group, there are $\delta$ local parity symbols.
    \item The remaining $h$ global parity symbols may depend on all $k$ symbols.
    \item For any set $E \subseteq [n]$ where $E$ is obtained by picking $\delta$ coordinates from each $\frac{k}{r}$ local groups, puncturing $\mathcal{C}$ in coordinates in $E$ yields a $[k+h, k]$ MDS code. 
\eit
\end{defn}
$[k, r, h, \delta]$ data-local MRC is optimum with respect to minimum distance bound in \eqref{eq:dmin_rdelta}. The minimum distance of a $[k, r, h, \delta]$ data-local MRC is given by 
\begin{equation} \label{eq:dminDLMRC}
d_{min} = h+\delta+1.
\end{equation}

\begin{defn}[Local Maximally Recoverable Code]\label{defn:local}
Let $C$ be a systematic $[n,k,d_{min}]$ code. We say that $\mathcal{C}$ is an $[k, r, h, \delta]$ local maximally recoverable code if the following conditions are satisfied
\bit
    \item $r | (k+h)$ and $n=k+\frac{k+h}{r}\cdot \delta+h$
    \item There are $k$ data symbols and $h$ global parity symbols where each global parity may depend on all data symbols.  
    \item These $k+h$ symbols are partitioned into $\frac{k+h}{r}$ groups of size $r$. For each group there are $\delta$ local parity symbols.
    \item For any set $E \subseteq [n]$ where $E$ is obtained by picking $\delta$ coordinates from each $\frac{k+h}{r}$ local groups, puncturing $\mathcal{C}$ in coordinates in $E$ yields a $[k+h, k]$ MDS code.
\eit
\end{defn}
$[k, r, h, \delta]$ local MRC is optimum with respect to minimum distance bound in \eqref{eq:dmin_rdelta}. The minimum distance of a $[k, r, h, \delta]$ local MRC is given by  
\begin{equation} \label{eq:dminLMRC}
d_{min} = h+ \delta + 1 + \floor[\Big]{\frac{h}{r}}\cdot \delta.
\end{equation}


Maximally recoverable codes with locality are also known in literature as Partial-MDS codes (PMDS) codes \cite{pmds}.  Constructions of PMDS codes with two and three global parities have been discussed in \cite{blaum2016construction, chen2015sector}. A general construction of PMDS codes based on linearized polynomials has been provided in \cite{calis2017general}. An improved construction of PMDS codes for all parameters over small field sizes has been presented in \cite{small}. Construction of MRCs  over small field sizes have been investigated in \cite{hu2016new, guruswami2018constructions}. Recently, construction of MRCs based on linearized Reed Solomon codes and skew polynomials have been studied in \cite{martinez2019universal, cai2020construction, gopi2020improved}.

\subsection{Codes with Hierarchical Locality}
The concept of \emph{locality} has been extended to hierarchical locality in \cite{hlocality}. In the case of $(r, \delta)$ locality, if there are more than $\delta$ erasures, then the code offers no locality. In the case of codes with hierarchical locality, the locality constraints are such that with the increase in the number of erasures, the locality increases in steps. The following is the definition of code with two-level hierarchical locality.

\begin{defn} \label{defn:hier_local}
An $[n, k, d_{min}]$ linear code $\mathcal{C}$ is a code with \emph{hierarchical locality} having parameters $[(r_1, \delta_1), (r_2, \delta_2)]$  if for every symbol $c_i$, $1 \leq i \leq n$, there exists a punctured code $\mathcal{C}_i$ such that $c_i \in Supp(C_i)$ and the following conditions hold, 
\begin{itemize}
    \item $dim(\mathcal{C}_i) \leq r_1$
    \item $d_{min}(\mathcal{C}_i) \geq \delta_1$ and
    \item $\mathcal{C}_i$ is a code with $(r_2, \delta_2)$ locality.
\end{itemize}
\end{defn}

An upper bound on the minimum distance of a code with two-level hierarchical locality is given by 
\begin{equation} \label{eqn:min_dist_bound}
    d \leq n - k + 1 -(\ceil[\Big]{\frac{k}{r_2}}-1)(\delta_2 - 1) - (\ceil[\Big]{\frac{k}{r_1}}-1)(\delta_1 - \delta_2).
\end{equation}

\subsection{Our Contributions}

In this work, we consider the locality constraints imposed by codes with two-level hierarchical locality and define maximally recoverable codes with data-local and local hierarchical locality. We prove that certain punctured codes of these codes are data-local/local MRCs. We derive the minimum distance of hierarchical data-local MRCs. We give a procedure to construct hierarchical data-local MRCs from hierarchical local MRCs. We provide a construction of hierarchical local MRCs for all parameters. We also give constructions of MRC with hierarchical locality for some parameters, whose field size is smaller than that of known constructions for general parameters.  We also derive a field size lower bound on MRC with hierarchical locality.

\subsection{Notation}
For any integer $n$, $[n] = \{1, 2, 3 \ldots, n\}$. For any $E \subseteq [n]$, $\bar{E} = [n]-E$. For any $[n,k]$ code, and any $E \subseteq [n]$, $\mathcal{C}|_E$ refers to the punctured code obtained by restricting $\mathcal{C}$ to the coordinates in $E$. This results in an $[n-|E|, k']$ code where $k' \leq k$. For any $m\times n$ matrix $H$ and $E \subseteq [n]$, $H|_E$ is the $m \times |E|$ matrix formed by restricting $H$ to columns indexed by $E$.
In several definitions to follow, we implicitly assume certain divisibility conditions which will be clear from the context.

\section{Maximally Recoverable Codes with Hierarchical Locality} \label{sec:HLMRC}

In this section, we define hierarchical data-local and local MRCs and illustrate the definitions through an example.

\begin{defn}[Hierarchical Data Local Code] \label{defn:HDLC}
We define a $ [k, r_1, r_2, h_1, h_2, \delta] $  hierarchical data local (HDL) code of length  \[n=k+h_1+\frac{k}{r_1}(h_2+\frac{r_1}{r_2}\delta)\] as follows:
\begin{itemize}
\item The code symbols $c_1, \ldots, c_n$ satisfy $h_1$ global parities given by $\sum_{j = 1}^n u_j^{(\ell)} c_j = 0, \ \ 1 \leq  \ell \leq h_1$.
\item The first $n-h_1$ code symbols are partitioned into $t_1 = \frac{k}{r_1}$ groups $A_i, 1 \leq i \leq t_1$ such that $|A_i| =  r_1 + h_2 + \frac{r_1}{r_2} \delta = n_1$. The code symbols in the $i^\text{th}$ group, $1 \leq i \leq t_1$ satisfy the following $h_2$  mid-level parities $\sum_{j = 1}^{n_1} v_{i,j}^{(\ell)} c_{(i-1)n_1+j} = 0, \ \ 1 \leq  \ell \leq h_2$.
\item The first $n_1-h_2$ code symbols of the $i^\text{th}$ group, $1 \leq i \leq t_1$ are partitioned into $t_2 = \frac{r_1}{r_2}$ groups $B_{i,s}, 1 \leq i \leq t_1, 1 \leq s \leq t_2$ such that $|B_{i,s}| =  r_2 + \delta = n_2$. The code symbols in the $(i,s)^\text{th}$ group, $1 \leq i \leq t_1, 1 \leq s \leq t_2$ satisfy the following $\delta$  local parities $\sum_{j = 1}^{n_2} w_{i,s,j}^{(\ell)} c_{(i-1)n_1+(s-1)n_2+j} = 0, \ \ 1 \leq  \ell \leq \delta$.
\end{itemize}
\end{defn}

\begin{defn}[Hierarchical Data Local MRC] \label{defn:HDLMRC}
Let $\mathcal{C}$ be a $[k, r_1, r_2, h_1, h_2, \delta]$ HDL code. Then $C$ is maximally recoverable if for any set $E \subset [n]$ such that $|E| = k+h_1$ and
\begin{enumerate}
    \item $E \bigcap B_{i, s} \leq r_2$ $\forall \ i, s$,
    \item $E \bigcap A_i  = r_1$ $\forall \ i$,
\end{enumerate}
the punctured code $\mathcal{C}|_E$ is a $[k+h_1,k,h_1+1]$ MDS code. 
\end{defn}

\begin{defn}[Hierarchical Local Code] \label{defn:HLC}
We define a $ [k, r_1, r_2, h_1, h_2, \delta] $ hierarchical local (HL) code of length  $n=k+h_1+\frac{k+h_1}{r_1}(h_2+\frac{r_1+h_2}{r_2}\delta)$ as follows:
\begin{itemize}
\item The code symbols $c_1, \ldots, c_n$ satisfy $h_1$ global parities given by $\sum_{j = 1}^n u_j^{(\ell)} c_j = 0, \ \ 1 \leq  \ell \leq h_1$.
\item The $n$ code symbols are partitioned into $t_1 = \frac{k+h_1}{r_1}$ groups $A_i, 1 \leq i \leq t_1$ such that $|A_i| = r_1 + h_2 + \frac{r_1+h_2}{r_2}\delta = n_1$. The code symbols in the $i^\text{th}$ group, $1 \leq i \leq t_1$ satisfy the following $h_2$  mid-level parities $\sum_{j = 1}^{n_1} v_{i,j}^{(\ell)} c_{(i-1)n_1+j} = 0, \ \ 1 \leq  \ell \leq h_2$.
\item The $n_1$ code symbols of the $i^\text{th}$ group, $1 \leq i \leq t_1$ are partitioned into $t_2 = \frac{r_1+h_2}{r_2}$ groups $B_{i,s}, 1 \leq i \leq t_1, 1 \leq s \leq t_2$ such that $|B_{i,s}| =  r_2 + \delta = n_2$. The code symbols in the $(i,s)^\text{th}$ group, $1 \leq i \leq t_1, 1 \leq s \leq t_2$ satisfy the following $\delta$  local parities $\sum_{j = 1}^{n_2} w_{i,s,j}^{(\ell)} c_{(i-1)n_1+(s-1)n_2+j} = 0, \ \ 1 \leq  \ell \leq \delta$.
\end{itemize}
\end{defn}

\begin{defn}[Hierarchical Local MRC] \label{defn:HLMRC}
Let $\mathcal{C}$ be a $[k, r_1, r_2, h_1, h_2, \delta]$ HL code. Then $C$ is maximally recoverable if for any set $E \subset [n]$ such that $|E| = k+h_1$ and
\begin{enumerate}
    \item $E \bigcap B_{i, s} \leq r_2$ $\forall \ i, s$,
    \item $E \bigcap A_i  = r_1$ $\forall \ i$,
\end{enumerate}
the punctured code $\mathcal{C}|_E$ is a $[k+h_1,k,h_1+1]$ MDS code. 
\end{defn}

In an independent parallel work \cite{martinez2019universal}, a class of MRCs known as multi-layer MRCs have been introduced. We would like to note that hierarchical local MRCs (given in Definition \ref{defn:HLMRC}) form a subclass of these multi-layer MRCs. One key difference between the codes constructed in \cite{martinez2019universal} and the current paper is that the authors in \cite{martinez2019universal} take the generator matrix based approach and we take the parity-check matrix based approach.

\begin{example}
We demonstrate the structure of the parity check matrix for an $[k=5, r_1=3, r_2=2, h_1=1, h_2=1 , \delta=2]$ HL code.
The length of the code is $n=k+h_1+\frac{k+h_1}{r_1}(h_2+\frac{r_1+h_2}{r_2}\delta) = 16$. The parity check matrix of the code is given below:


\[
	H = \begin{bmatrix}
	\begin{matrix}
		\begin{matrix}
		M_{1,1} \\
		& M_{1,2} \\
		\end{matrix} \\
	
		N_1 \\
	
		& \begin{matrix}
			M_{2,1} \\
			& M_{2,2} \\
		  \end{matrix} \\
	
		& N_2 \\
	\end{matrix}\\
	
	O \\
	\end{bmatrix}
\]

where,

\[
	M_{i,j} = \begin{bmatrix}
	w_{i,j,1}^{(1)} & w_{i,j,2}^{(1)} & w_{i,j,3}^{(1)} & w_{i,j,4}^{(1)} \\
	w_{i,j,1}^{(2)} & w_{i,j,2}^{(2)} & w_{i,j,3}^{(2)} & w_{i,j,4}^{(2)} \\
	\end{bmatrix}
\]
\[
	N_i = \begin{bmatrix}
	v_{i,1}^{(1)} & v_{i,2}^{(1)} & \ldots & v_{i,8}^{(1)} \\
	\end{bmatrix}
\]
\[
	O = \begin{bmatrix}
	u_1^{(1)} & u_2^{(1)} & \ldots & u_{16}^{(1)} \\
	\end{bmatrix}
\]

\end{example}


\begin{table*}[ht]
	\centering
		\begin{tabular}{|c|c|c|}
			\hline
			Reference    &  Parameters of HL-MRC & Field Size  \\
			\hline
			Construction 1 in \cite{martinez2019universal}     &   $[k,r_1,r_2,h_1,h_2, \delta]$     & $O(t_1^{r_1})$ \\
			\hline
			Construction \ref{constr:HLMRC} &    $[k,r_1,r_2,h_1,h_2, \delta]$ & $O(n_2 n_1^{(\delta+1)h_1-1}n^{(\delta+1)(h_2+1)h_1 - 1})$ \\
			in the current paper & & \\
			\hline
			Construction  \ref{constr:h1}  &  $[k,r_1,r_2,h_1=1,h_2, \delta]$   & $O(n_2 n_1^{(\delta+1)(h_2+1)-1})$\\
			in the current paper & & \\
			\hline
			Construction  \ref{constr:h11h21}  &  $[k,r_1,r_2,h_1=1,h_2=1, \delta]$   & $O(n_1)$\\
			in the current paper & & \\
			\hline			
			Construction  \ref{constr:h12h21}  &  $[k,r_1,r_2,h_1=2,h_2=1, \delta]$   & $O(n^4)$\\
			in the current paper & & \\
			\hline
		\end{tabular}
	\vspace*{0.2in}
	\caption{ Summary of HL-MRC constructions.}
	\label{tab:performance}
\end{table*}

\section{Properties of MRC with Hierarchical Locality}

In this section, we will derive two properties of MRC with hierarchical locality. We will show that the middle codes of a HDL/HL-MRC have to be data-local and local MRC respectively. Also, we derive the minimum distance of HDL MRC.

\begin{lem} \label{lem:midHDLMRC}
Consider a $[k, r_1, r_2, h_1, h_2, \delta]$ HDL-MRC $\mathcal{C}$. Let $A_i, 1 \leq i \leq t_1$ be the supports of the middle codes as defined in Definition \ref{defn:HDLC}. Then, for each $i$, $\mathcal{C}_{A_i}$ is a $[r_1, r_2, h_2, \delta]$ data-local MRC.
\end{lem}
\begin{proof}
Suppose not. This means that for some $i$, the middle code $\mathcal{C}_{A_i}$ is not a $[r_1, r_2, h_2, \delta]$ data-local MRC. By the definition of data-local MRC, we have that there exists a set $E_1 \subset A_i$ such that $|E_1| = r_1 + h_2$ and $\mathcal{C}_{E_1}$ is not an $[r_1+h_2, r_2, h_2+1]$ MDS code. This implies that there exists a subset $E' \subset E_1$ such that $|E'| = r_1$ and $\text{rank}(G|_{E'}) < r_1$. We can extend the set $E'$ to obtain a set $E \subset [n]$, $|E| = k+h_1$ which satisfies the conditions in the definition of HDL-MRC. The resulting punctured code $\mathcal{C}_{E}$ cannot be MDS since there exists an $r_1 < k$ sized subset of $E$ such that $\text{rank}(G|_{E'}) < r_1$.
\end{proof}

\begin{lem} \label{lem:midHLMRC}
Consider a $[k, r_1, r_2, h_1, h_2, \delta]$ HL-MRC $\mathcal{C}$. Let $A_i, 1 \leq i \leq t_1$ be the supports of the middle codes as defined in Definition \ref{defn:HLC}. Then, for each $i$, $\mathcal{C}_{A_i}$ is a $[r_1, r_2, h_2, \delta]$ local MRC.
\end{lem}
\begin{proof}
Proof is similar to the proof of Lemma \ref{lem:midHDLMRC}.
\end{proof}

\subsection{Minimum Distance of HDL-MRC}

\begin{lem}
The minimum distance of a $[k, r_1, r_2, h_1, h_2, \delta]$ HDL-MRC is given by $d = h_1 + h_2 + \delta +1$.
\end{lem}
\begin{proof}
Based on the definition of HDL-MRC, it can be seen that the $[k, r_1, r_2, h_1, h_2, \delta]$ HDL-MRC is a code with hierarchical locality as per Definition \ref{defn:hier_local} with the following parameters:
\begin{itemize}
\item $k, r_1,r_2$ are the same.
\item $ \delta_2 -1 = \delta$, $\delta_1  =  h_2 + \delta +1$.
\item $n  =  k + h_1 + \frac{k}{r_1}(h_2 + \frac{r_1}{r_2}\delta)$.
\end{itemize}
Substituting these parameters in the minimum distance bound in \eqref{eqn:min_dist_bound}, we have that
\begin{equation}
d \leq h_1 + h_2 + \delta +1.
\end{equation}

By Lemma \ref{lem:midHDLMRC}, we know that $\mathcal{C}_{A_i}$ is a $[r_1, r_2, h_2, \delta]$ data-local MRC. The minimum distance of $\mathcal{C}_{A_i}$ (from \eqref{eq:dminLMRC}) is $h_2 + \delta +1$. Thus, the middle code itself can recover from any $h_2 + \delta$ erasures. The additional $h_1$ erasures can be shown to be extended to a set $E$ (consisting of $k$ additional non-erased symbols) which satisfies the conditions in Definition \ref{defn:HDLMRC}. Since, the punctured code $\mathcal{C}|_E$ is a $[k+h_1,k,h_1+1]$ MDS code, it can be used to recover the $h_1$ erasures. Hence, $[k, r_1, r_2, h_1, h_2, \delta]$ HDL-MRC can recover from any $h_1 + h_2 + \delta$ erasures.
\end{proof}
 

\subsection{Deriving HDL-MRC from HL-MRC}
In this section, we give a method to derive any HDL-MRC from a HL-MRC. Assume an $ [k, r_1, r_2, h_1, h_2, \delta] $ HL-MRC $\mathcal{C}$. Consider a particular set $E$ of $k+h_1$ symbols satisfying the conditions given in Definition \ref{defn:HLMRC}. We will refer to the elements of set $E$ as ``primary symbols". By the definition of HL-MRC, the code $\mathcal{C}$ when punctured to $E$ results in a $[k+h_1, k, h_1+1]$ MDS code. Hence, any $k$ subset of $E$ forms an information set. We will refer to the first $k$ symbols of $E$ as ``data symbols" and the rest $h_1$ symbols as global parities.
The symbols in $[n]\setminus E$ will be referred to as parity symbols (mid-level parities and local parities) and it can be observed that the parity symbols can be obtained as linear combinations of data symbols.
\begin{itemize}
    \item If $r_1 \mid h_1$ and $r_2 \mid h_2$
    \begin{enumerate}
        \item For $A_i, \frac{k}{r_1} < i \leq \frac{k+h_1}{r_1}$, drop all the parity symbols, including $h_2$ mid-level parities per $A_i$ as well as the $\delta$ local parities per $B_{i,s} \subset A_i$. As a result, we would be left with $h_1$ ``primary symbols" in the local groups $A_i, \frac{k}{r_1} < i \leq \frac{k+h_1}{r_1}$. These form the global parities of the HDL-MRC. This step ensures that mid-level and local parities formed from global parities are dropped.
        
                \item For each $B_{i, s},\: 1 \leq i \leq \frac{k}{r_1},\: s > \frac{r_1}{r_2}$, drop the $\delta$ local parities. This step ensures that local parities formed from mid-level parities are dropped.

    \end{enumerate}
    This results in an $ [k, r_1, r_2, h_1, h_2, \delta] $ HDL-MRC.
    
    \item If $r_1 \nmid h_1$ and $r_2 \mid h_2$,
    \begin{enumerate}
        \item From the groups $A_i, \lfloor\frac{k}{r_1}\rfloor + 1 < i \leq \frac{k+h_1}{r_1}$, drop all the parity symbols, including $h_2$ mid-level parities per $A_i$ as well as the $\delta$ local parities per $B_{i,s} \subset A_i$.
        \item For each $B_{i, s},\: 1 \leq i \leq \lfloor\frac{k}{r_1}\rfloor,\: s > \frac{r_1}{r_2}$, drop the $\delta$ local parities.

        \item Drop the $k - \lfloor\frac{k}{r_1}\rfloor r_1$ data symbols in $A_i, i = \lfloor\frac{k}{r_1}\rfloor + 1$ and recalculate all the parities (local, mid-level and global) by setting these data symbols as zero in the linear combinations.
    \end{enumerate}
    This results in an $ [\lfloor\frac{k}{r_1}\rfloor r_1, r_1, r_2, h_1, h_2, \delta] $ HDL-MRC.  
\end{itemize}
For the case of $r_2 \nmid h_2$, HDL-MRC can be derived from HL-MRC using similar techniques as above. Hence, in the rest of the paper, we will discuss the constructions of HL-MRC.

\section{General Construction of HL-MRC}

In this section, we will present a general construction of $ [k, r_1, r_2, h_1, h_2, \delta]$ HL-MRC. First, we will provide the structure of the code and then derive necessary and sufficient conditions for the code to be HL-MRC. Finally, we will apply a known result of BCH codes to complete the construction.

\begin{defn}
	A multiset $S \subseteq \mathbb{F}$ is $k$-wise independent over $\mathbb{F}$ if for every set $T \subseteq S$ such that $|T| \leq k$, $T$ is linearly independent over $\mathbb{F}$.
\end{defn}
\begin{lem} \label{lem:mds_linearized}
Let $\mathbb{F}_{q^m}$ be an extension of $\mathbb{F}_q$. Let $a_1, a_2, \ldots, a_n$ be elements of $\mathbb{F}_{q^m}$. The following matrix
\[
\begin{bmatrix}
a_1 & a_2 & a_3 & \ldots & a_n \\
a_1^{q} & a_2^{q} & a_3^{q} & \ldots & a_n^{q} \\
\vdots & \vdots & \vdots & \ldots & \vdots \\
a_1^{q^{k-1}} & a_2^{q^{n-1}} & a_3^{q^{n-1}} & \ldots & a_n^{q^{k-1}} \\
\end{bmatrix}
\]
is the generator matrix of a $[n,k]$ MDS code if and only if $a_1, a_2, \ldots, a_n$ are $k$-wise linearly independent over $\mathbb{F}_q$.
\end{lem}
\begin{proof}
Directly follows from Lemma 3 in \cite{small}.
\end{proof}

\begin{constr} \label{constr:HLMRC}
The structure of the parity check matrix of a $ [k, r_1, r_2, h_1, h_2, \delta]$ HL-MRC is given by
\[
H = \begin{bmatrix}
    H_0   & \\
          & H_0 & \\
          &     & \ddots & \\
          &     &        &  H_0 \\
    H_1  & H_2  & \hdots &  H_{t_1} \\
    \end{bmatrix}
\]
Here, $H_0$ is an $(t_2 \cdot \delta  + h_2)\times n_1 $ matrix and $H_i, 1\leq i \leq t_1$ are an $h_1\times n_1$ matrix. $H_0$ is then further subdivided as follows:

\[
H_0 = \begin{bmatrix}
    M_0   & \\
          & M_0 & \\
          &     & \ddots & \\
          &     &        &  M_0 \\
    M_1  & M_2  & \hdots &  M_{t_2} \\
        \end{bmatrix}
\]
$M_0$ has the dimensions $\delta \times n_2$ and $M_i, 1\leq i \leq t_2$ is an $h_2 \times n_2$ matrix.

Assume $q$ to be a prime power such that $q \geq n$, $\mathbb{F}_{q^{M_1}}$ be an extension field of $\mathbb{F}_q$ and $\mathbb{F}_{q^M}$ is an extension field of $\mathbb{F}_{q^{M_1}}$, where $M_1 \mid M$. 

In this case, the construction is given by the following.
\[
M_0 = \begin{bmatrix}
        1 & 1 & 1 & \ldots & 1 \\
        0 & \beta & \beta^2 & \ldots & \beta^{n_2 -1} \\
        0 & \beta^2 & \beta^4 & \ldots & \beta^{2(n_2-1)} \\
        \vdots & \vdots & \vdots & \ldots & \vdots \\
        0 & \beta^{\delta-1} & \beta^{2(\delta - 1)} & \ldots & \beta^{(\delta - 1)(n_2-1)}
        \end{bmatrix},
\]
where $\beta \in \mathbb{F}_q$ is a primitive element.
\[
M_i = \begin{bmatrix}
        \alpha_{i, 1} & \alpha_{i, 2} & \ldots & \alpha_{i, n_2} \\
        \alpha_{i, 1}^{q} & \alpha_{i, 2}^{q} & \ldots & \alpha_{i, n_2}^{q} \\
        \vdots & \vdots & \ldots & \vdots \\
        \alpha_{i, 1}^{q^{h_2-1}} &  \alpha_{i, 2}^{q^{h_2-1}} & \ldots &  \alpha_{i, n_2}^{q^{h_2-1}} \\
        \end{bmatrix},
\]
where $i \in [t_2]$, $ \alpha_{i,j} \in \mathbb{F}_{q^{M_1}}, 1 \leq i \leq t_2, 1 \leq j \leq n_2$.
\[
H_i = [H_{i,1} \ H_{i,2} \ldots H_{i,t_2}]
\]
\[
H_{i,s} = \begin{bmatrix}
        \lambda_{i,s,1} & \lambda_{i,s, 2} & \ldots & \lambda_{i, s,n_2} \\
        \lambda_{i, s,1}^{q^{M_1}} & \lambda_{i, s,2}^{q^{M_1}} & \ldots & \lambda_{i, s,n_2}^{q^{M_1}} \\
        \vdots & \vdots & \ldots & \vdots \\
        \lambda_{i, s,1}^{q^{M_1(h_1-1)}} &  \lambda_{i, s,2}^{q^{M_1(h_1-1)}} & \ldots &  \lambda_{i, s,n_2}^{q^{M_1(h_1-1)}} \\
        \end{bmatrix},
\]

where $i \in [t_1], s \in [t_2]$, $ \lambda_{i,s,j} \in \mathbb{F}_{q^{M}}, 1 \leq i \leq t_1, 1 \leq s \leq t_2, 1 \leq j \leq n_2$.


\end{constr}

A $(\delta,h_2)$ erasure pattern is defined by the following two sets: 
\begin{itemize}
\item $\Delta$ is a three dimensional array of indices with the first dimension $i$ indexing the middle code and hence $1 \leq I \leq t_1$, the second dimension $s$ indexing the local code and hence $1 \leq s \leq t_2$. The third dimension $j$ varies from $1$ to $\delta$ and used to index the $\delta$ coordinates which are erased in the $(i,s)^{\text{th}}$ group. Let $e \in [n]$ denote the actual index of the erased coordinate in the code and $e \in B_{i,s}$, then we set $\Delta_{i,s,j} = (e \mod n_2) + 1$. $\Delta_{i,s}$ is used to denote the vector of $\delta$ coordinates which are erased in the $(i,s)^{\text{th}}$ group. $\bar{\Delta}_{i,s}$ is used to denote the complement of $\Delta_{i,s}$ in the set $[n_2]$.
\item $\Gamma$ is a two dimensional array of indices with the first dimension $i$ indexing the middle code  and hence $1 \leq i \leq t_1$. The second dimension $j$ varies from $1$ to $h_2$ and used to index the additional $h_2$ coordinates which are erased in the $i^{\text{th}}$ group. Let $e \in [n]$ denote the actual index of the erased coordinate in the code and $e \in A_i$, then we set $\Gamma_{i,j} = (e \mod n_1) + 1$. $\Gamma_{i}$ is used to denote the vector of $h_2$ coordinates which are erased in the $i^{\text{th}}$ group. $\bar{\Gamma}_{i}$ is used to denote the complement of $\Gamma_{i}$ in the set $[n_1] \setminus (\cup_{s=1}^{t_2} \Delta_{i,s})$.
\end{itemize}
We define some matrices and sets based on the parameters of the construction, which will be useful in proving the subsequent necessary and sufficient condition for the construction to be HL-MRC.
\begin{eqnarray*}
L_{i,s} & = & (M_0|_{\Delta_{i,s}})^{-1}  M_0|_{\bar{\Delta}_{i,s}} \\
\Psi_i  & = & \{ \alpha_{s,\bar{\Delta}_{i,s}} + \alpha_{s,\Delta_{i,s}} L_{i,s}, 1 \leq s \leq t_2 \} \\
& = & \{ \Psi_{i,\Gamma_i}, \  \Psi_{i,\bar{\Gamma}_i} \} \\
& = & \{ \psi_{i,1}, \ldots, \psi_{i,h_2}, \psi_{i,h_2+1}, \ldots, \psi_{i, r_1+h_2} \}
\end{eqnarray*}
The above equalities follow by noting that the $\cup_{s=1}^{t_2} \bar{\Delta}_{i,s} = \Gamma_i \cup \bar{\Gamma}_i$.
We will refer to the elements in $\Psi_{i,\Gamma_i}$ by $ \{ \psi_{i,1}, \ldots, \psi_{i,h_2}\}$ and those in $\Psi_{i,\bar{\Gamma}_i}$ by $\{ \psi_{i,h_2+1}, \ldots, \psi_{i, r_1+h_2}\}$.
Consider the following matrix based on the elements of $\Psi_i$,
\begin{equation}
F_i = [F_i|_{\Gamma_i} \ \ F_i|_{\bar{\Gamma}_i}] = \begin{bmatrix}
        \psi_{i, 1} & \psi_{i, 2} & \ldots & \psi_{i, r_1+h_2} \\
        \psi_{i, 1}^{q} & \psi_{i, 2}^{q} & \ldots & \psi_{i, r_1+h_2}^{q} \\
        \vdots & \vdots & \ldots & \vdots \\
        \psi_{i, 1}^{q^{h_2-1}} &  \psi_{i, 2}^{q^{h_2-1}} & \ldots &  \psi_{i, r_1+h_2}^{q^{h_2-1}} \\
        \end{bmatrix},
\end{equation}
And
\begin{eqnarray*}
\Phi_i  & = & \{ \lambda_{i,s,\bar{\Delta}_{i,s}} + \lambda_{i,s,\Delta_{i,s}} L_{i,s}, 1 \leq s \leq t_2 \} \\
& = & \{ \Phi_{i,\Gamma_i}, \  \Phi_{i,\bar{\Gamma}_i} \} \\
& = & \{ \phi_{i,1}, \ldots, \phi_{i,h_2}, \phi_{i,h_2+1}, \ldots, \phi_{i, r_1+h_2} \}
\end{eqnarray*}

\begin{equation}
Z_i = (F_i|_{\Gamma_i})^{-1} F_i|_{\bar{\Gamma}_i}
\end{equation}
Finally, the set $\Theta = \{ \Phi_{i,\bar{\Gamma}_i} + \Phi_{i,\Gamma_i} Z_i, 1 \leq i \leq t_1 \}$.

\begin{thm} \label{thm:necsuf}
The code described in Construction \ref{constr:HLMRC} is a $[k, r_1, r_2, h_1, h_2, \delta]$ HL-MRC if and only if for any $(\delta,h_2)$ erasure pattern, the following two conditions are satisfied:
\begin{enumerate}
\item Each $\Psi_i, 1 \leq i \leq t_1$ is $h_2$-wise independent over $\mathbb{F}_q$.
\item $\Theta$ is $h_1$-wise independent over $\mathbb{F}_{q^{M_1}}$. 
\end{enumerate}
\end{thm}

\begin{proof}
By Lemma \ref{lem:midHLMRC}, we have that $\mathcal{C}$ is a HL-MRC if and only if the $\mathcal{C}|_{A_i}$ is a $[r_1, r_2, h_2, \delta]$ local MRC.
By the definition of local MRC, a code is a $[r_1, r_2, h_2, \delta]$ local MRC, if after puncturing $\delta$ coordinates in each of the $\frac{r_1+h_2}{r_2}$ local groups, the resultant code is $[r_1+h_2, r_1, h_2+1]$ MDS code. 

The puncturing on a set of coordinates in the code is equivalent to shortening on the same set of coordinates in the dual code. Shortening on a set of coordinates in the dual code can be performed by zeroing the corresponding coordinates in the parity check matrix by row reduction. To prove that  $\mathcal{C}|_{A_i}$ is a $[r_1, r_2, h_2, \delta]$ local MRC, we need to show that certain punctured codes are MDS (Definition \ref{defn:local}). We will equivalently that the shortened codes of the dual code are MDS.

Consider the coordinates corresponding to $(i,s)^\text{th}$ group in the parity check matrix. The sub-matrix of interest in this case is the following:
\begin{equation*}
\left [\begin{array}{c|c}
M_0|_{\Delta_{i,s}} & M_0|_{\bar{\Delta}_{i,s}} \\
\hline
\alpha_{s,\Delta_{i,s}} & \alpha_{s,\bar{\Delta}_{i,s}} \\
\alpha_{s,\Delta_{i,s}}^q & \alpha_{s,\bar{\Delta}_{i,s}}^q \\
\vdots & \vdots \\
\alpha_{s,\Delta_{i,s}}^{q^{h_2-1}} & \alpha_{s,\bar{\Delta}_{i,s}}^{q^{h_2-1}}
\end{array} \right],
\end{equation*} 
Where $\alpha_{s,\Delta_{i,s}}^q$ is the vector obtained by taking $q^{\text{th}}$ power of each element in the vector.
Applying row reduction to the above matrix, we have
\begin{equation*}
\left [\begin{array}{c|c}
M_0|_{\Delta_{i,s}} & M_0|_{\bar{\Delta}_{i,s}} \\
\hline
\bold{0} & \alpha_{s,\bar{\Delta}_{i,s}} + \alpha_{s,\Delta_{i,s}} L_{i,s}\\
 \bold{0} & (\alpha_{s,\bar{\Delta}_{i,s}} + \alpha_{s,\Delta_{i,s}} L_{i,s})^q \\
\vdots & \vdots \\
\bold{0} & (\alpha_{s,\bar{\Delta}_{i,s}} + \alpha_{s,\Delta_{i,s}} L_{i,s})^{q^{h_2-1}}
\end{array} \right]
\end{equation*} 
Note that $L_{i,s}$ can be pushed into the power of $q$ since the elements of $L_{i,s}$ are in $\mathbb{F}_q$. After row reducing $\delta$ coordinates from each of the $\frac{r_1+h_2}{r_2}$ local groups in $A_i$, the resultant parity check matrix is $F_i$. Applying Lemma \ref{lem:mds_linearized}, $F_i$ forms the generator matrix of an MDS code if and only if the set $\Psi_i$ is $h_2$-wise independent over $\mathbb{F}_q$.
The shortening of the code above is applicable to mid-level parities. Now, we will apply similar shortening in two steps to global parities. The sub-matrix of interest in this case is the following:
\begin{equation*}
\left [\begin{array}{c|c}
M_0|_{\Delta_{i,s}} & M_0|_{\bar{\Delta}_{i,s}} \\
\hline
\alpha_{s,\Delta_{i,s}} & \alpha_{s,\bar{\Delta}_{i,s}} \\
\alpha_{s,\Delta_{i,s}}^q & \alpha_{s,\bar{\Delta}_{i,s}}^q \\
\vdots & \vdots \\
\alpha_{s,\Delta_{i,s}}^{q^{h_2-1}} & \alpha_{s,\bar{\Delta}_{i,s}}^{q^{h_2-1}} \\
\hline
\lambda_{i,s,\Delta_{i,s}} & \lambda_{i,s,\bar{\Delta}_{i,s}} \\
\lambda_{i,s,\Delta_{i,s}}^{q^{M_1}} & \lambda_{i,s,\bar{\Delta}_{i,s}}^{q^{M_1}} \\
\vdots & \vdots \\
\lambda_{i,s,\Delta_{i,s}}^{q^{M_1(h_1-1)}} & \lambda_{i,s,\bar{\Delta}_{i,s}}^{q^{M_1(h_1-1)}}
\end{array} \right]
\end{equation*} 
Applying row reduction to the above matrix, we have
\begin{equation*}
\left [\begin{array}{c|c}
M_0|_{\Delta_{i,s}} & M_0|_{\bar{\Delta}_{i,s}} \\
\hline
\bold{0} & \alpha_{s,\bar{\Delta}_{i,s}} + \alpha_{s,\Delta_{i,s}} L_{i,s}\\
 \bold{0} & (\alpha_{s,\bar{\Delta}_{i,s}} + \alpha_{s,\Delta_{i,s}} L_{i,s})^q \\
\vdots & \vdots \\
\bold{0} & (\alpha_{s,\bar{\Delta}_{i,s}} + \alpha_{s,\Delta_{i,s}} L_{i,s})^{q^{h_2-1}} \\
\hline
\bold{0} & \lambda_{i,s,\bar{\Delta}_{i,s}} + \lambda_{i,s,\Delta_{i,s}} L_{i,s}\\
 \bold{0} & (\lambda_{i,s,\bar{\Delta}_{i,s}} + \lambda_{i,s,\Delta_{i,s}} L_{i,s})^{q^{M_1}} \\
\vdots & \vdots \\
\bold{0} & (\lambda_{i,s,\bar{\Delta}_{i,s}} + \lambda_{i,s,\Delta_{i,s}} L_{i,s})^{q^{M_1(h_1-1)}}
\end{array} \right]
\end{equation*} 
To apply row reduction again, we consider the following submatrix obtained by deleting the zero columns and aggregating the non-zero columns from the $\frac{r_1+h_2}{r_2}$ groups, 
\begin{equation*}
\left [\begin{array}{c|c}
F_i|_{\Gamma_i} & F_i|_{\bar{\Gamma}_i} \\
\hline
\Phi_{i,\Gamma_i} & \Phi_{i,\bar{\Gamma}_i} \\
\Phi_{i,\Gamma_i}^{q^{M_1}} & \Phi_{i,\bar{\Gamma}_i}^{q^{M_1}} \\
\vdots & \vdots \\
\Phi_{i,\Gamma_i}^{q^{M_1(h_1-1)}} & \Phi_{i,\bar{\Gamma}_i}^{q^{M_1(h_1-1)}}
\end{array} \right]
\end{equation*}
Applying row reduction to the above matrix, we have
\begin{equation*}
\left [\begin{array}{c|c}
F_i|_{\Gamma_i} & F_i|_{\bar{\Gamma}_i} \\
\hline
\bold{0} & \Phi_{i,\bar{\Gamma}_i} + \Phi_{i,\Gamma_i} Z_i \\
\bold{0} & (\Phi_{i,\bar{\Gamma}_i} + \Phi_{i,\Gamma_i} Z_i)^{q^{M_1}} \\
\vdots & \vdots \\
\bold{0} & (\Phi_{i,\bar{\Gamma}_i} + \Phi_{i,\Gamma_i} Z_i)^{q^{M_1(h_1-1)}}
\end{array} \right]
\end{equation*}
Note that $Z_i$ can be pushed into the power of $q^{M_1}$ since the elements of $Z_i$ are in $\mathbb{F}_{q^{M_1}}$. Applying Lemma \ref{lem:mds_linearized}, the row reduced matrix above forms the generator matrix of an MDS code if and only if the set $\Theta$ is $h_1$-wise independent over $\mathbb{F}_{q^{M_1}}$.
 \end{proof}
 
 \begin{lem} \label{lem:hwise_ind}
 For any $(\delta,h_2)$ erasure pattern,
 \begin{itemize}
 \item For each $i$, $\Psi_i = \{ \alpha_{s,\bar{\Delta}_{i,s}} + \alpha_{s,\Delta_{i,s}} L_{i,s}, 1 \leq s \leq t_2 \}$ is $h_2$-wise independent over $\mathbb{F}_q$ if the set 
 \begin{equation*}
 \{ \alpha_{s,j}, 1 \leq s \leq t_2, 1 \leq j \leq n_2 \}
 \end{equation*}
 is $(\delta+1) h_2$-wise independent over $\mathbb{F}_q$.
 \item $\Theta = \{ \Phi_{i,\bar{\Gamma}_i} + \Phi_{i,\Gamma_i} Z_i, 1 \leq i \leq t_1 \}$ is $h_1$-wise independent over $\mathbb{F}_{q^{M_1}}$ if the set 
 \begin{equation*}
 \{ \lambda_{i,s,j}, 1 \leq i \leq t_1, 1 \leq s \leq t_2, 1 \leq j \leq n_2 \}
 \end{equation*}
 is $(\delta+1)(h_2+1)h_1$-wise independent over $\mathbb{F}_{q^{M_1}}$.
 \end{itemize}
 
 \end{lem}
 \begin{proof}
 Since the size of matrix $L_{i,s}$ is $\delta \times (n_2 - \delta)$, each element of $\Psi_i$ can be a $\mathbb{F}_q$-linear combination of atmost $\delta + 1$ different $\alpha_{s,j}$. Consider $\mathbb{F}_q$-linear combination of $h_2$ elements in $\Psi_i$.  The linear combination will have at most $(\delta + 1) h_2$ different $\alpha_{s,j}$. Thus, if the set $\{ \alpha_{s,j} \}$ is $(\delta+1) h_2$-wise independent over $\mathbb{F}_q$, then $\Psi_i$ is $h_2$-wise independent over $\mathbb{F}_q$. To prove the second part, we note that each element of $\Phi_i$ is a linear combination of at most $\delta + 1$ different $\lambda_{i,s,j}$. Since the size of the matrix $Z_i$ is $h_2 \times (n_1 - h_2)$, each element of $\Theta$ can be a $\mathbb{F}_{q^{M_1}}$-linear combination of atmost $(\delta+1)(h_2+1)$ different $\lambda_{i,s,j}$. Consider $\mathbb{F}_{q^{M_1}}$-linear combination of $h_1$ elements in $\Theta$.  The linear combination will have at most $(\delta + 1) (h_2+1)h_1$ different $\lambda_{i,s,j}$. Thus, if the set $\{ \lambda_{i,s,j} \}$ is $(\delta + 1) (h_2+1)h_1$-wise independent over $\mathbb{F}_{q^{M_1}}$, then $\Theta$ is $h_1$-wise independent over $\mathbb{F}_{q^{M_1}}$.
 \end{proof}
 
 We will design the $\{ \alpha_{s,j} \}$ and $\{ \lambda_{i,s,j} \}$ based on the Lemma \ref{lem:hwise_ind} so that the field size is minimum possible. We will pick these based on the following two properties:
 \begin{itemize}
 \item {\bf Property 1:} The columns of parity check matrix of an $[n,k,d]$ linear code over $\mathbb{F}_q$ can be interpreted as $n$ elements over $\mathbb{F}_{q^{n-k}}$ which are $(d-1)$-wise linear independent over $\mathbb{F}_q$.
 \item {\bf Property 2:} There exists $[n = q^m-1, k,d]$ BCH codes over $\mathbb{F}_q$ \cite{roth2006}, where the parameters are related as
 \begin{equation*}
 n-k  = 1 + \left \lceil \frac{q-1}{q} (d-2) \right \rceil  \lceil \log_2 (n)\rceil 
 \end{equation*}
 \end{itemize}
 
 \begin{thm} \label{thm:HLBCH}
 The code in Construction \ref{constr:HLMRC} is a $ [k, r_1, r_2, h_1, h_2, \delta]$ HL-MRC if the parameters are picked as follows:
 \begin{enumerate}
\item $q$ is the smallest prime power greater than $n_2$.
\item $M_1$ is chosen based on the following relation:
\begin{equation*}
M_1 = 1 + \left \lceil \frac{q-1}{q} ((\delta+1)h_2-1) \right \rceil  \lceil \log_q (n_2 t_2)\rceil 
\end{equation*}
\item $n_2 t_2$ elements $\{ \alpha_{s,j} \}$ over $\mathbb{F}_{q^{M_1}}$ are set to be the columns of parity check matrix of the BCH code over $\mathbb{F}_q$ with parameters $[n = q^{\lceil \log_q (n_2 t_2) \rceil}-1, q^{\lceil \log_q (n_2 t_2)\rceil}-1-M_1, (\delta+1)h_2+1]$ .
\item $M$ is chosen to be the smallest integer dividing $M_1$ based on the following relation:
\begin{equation*}
M \geq 1 + \left \lceil \frac{q^{M_1}-1}{q^{M_1}} ((\delta+1)(h_2+1)h_1-1) \right \rceil \lceil \log_{q^{M_1}} (n)\rceil  
\end{equation*}
\item $n$ elements $\{ \lambda_{i,s,j} \}$ over $\mathbb{F}_{q^{M}}$ are set to be the columns of parity check matrix of the BCH code over $\mathbb{F}_{q^{M_1}}$ with parameters $[n = q^{M_1\lceil \log_{q^{M_1}} (n) \rceil}-1, q^{M_1\lceil \log_{q^{M_1}} (n) \rceil}-1-M, (\delta+1)(h_2+1)h_1+1]$ .
\end{enumerate}
 \end{thm}
\begin{proof}
The proof follows from Lemma \ref{lem:hwise_ind} and Properties 1 and 2.
\end{proof}
%
%
%
%
%
%

\section{Constructions of HL MRC over Low Field Sizes}

In this section, we give three constructions of HL MRC. The first one is for the case when $h_1=1$, the second construction is for the case when both $h_1 = 1$ and $h_2=1$ and the third one is for the case when $h_1=2$ and $h_2 = 1$.

  The second construction is for the case when $h_1=2$ and $h_2=1$. This construction is based on the construction of local MRC with $3$ global parities in \cite{gopi2020maximally}.

\subsection{HL-MRC Construction for $h_1 = 1$}

In this section, we present a construction of HL-MRC for the case when $h_1=1$ over a field size lower than that provided by Construction \ref{constr:HLMRC}.

\begin{constr}\label{constr:h1}
	The structure of the parity check matrix for the present construction is the same as that given in Construction \ref{constr:HLMRC}. In addition, the matrices $M_0$ and $M_i,\: 1 \leq i \leq t_2$ also remain the same. We modify the matrix $H_i, \; 1 \leq i \leq t_1$ as follows:
	\[ H_i = \begin{bmatrix}
	\alpha_{1, 1}^{q^{h_2}} & \alpha_{1, 2}^{q^{h_2}} & \ldots \alpha_{t_2, n_2}^{q^{h_2}}
	\end{bmatrix},
	\]
where $\{ \alpha_{s,j} \in \mathbb{F}_{q^{M_1}}, 1 \leq s \leq t_2, 1 \leq j \leq n_2 \}$ are chosen to be $(\delta+1) \cdot (h_2+1)$-wise independent over $\mathbb{F}_q$ based on Theorem \ref{thm:HLBCH}.

\end{constr}

\begin{thm}
The code $C$ given by Construction \ref{constr:h1} is a $ [k, r_1, r_2, h_1=1, h_2, \delta] $ HL-MRC.
\end{thm}

\begin{proof}
We show that $H$ can be used to correct all erasure patterns defined in Definition \ref{defn:HLMRC}. From the definition the code should recover from:
\begin{enumerate}
	\item $\delta$ errors per $B_{i,s}$ 
	\item $h_2$ additional errors per $A_i$
	\item $1$ more erasure anywhere in the entire code.
\end{enumerate}


Now, with $h_1=1$, the last erasure can be part of one group. Thus, effectively the code should recover from $h_2+1$ erasures per group. Suppose that the last erasure is in the $i^{\text{th}}$ group. The submatrix of interest for the $(i,s)^{th}$ local group is
\begin{equation*}
\left [\begin{array}{c|c}
M_0|_{\Delta_{i,s}} & M_0|_{\bar{\Delta}_{i,s}} \\
\hline
\alpha_{s,\Delta_{i,s}} & \alpha_{s,\bar{\Delta}_{i,s}} \\
\alpha_{s,\Delta_{i,s}}^q & \alpha_{s,\bar{\Delta}_{i,s}}^q \\
\vdots & \vdots \\
\alpha_{s,\Delta_{i,s}}^{q^{h_2-1}} & \alpha_{s,\bar{\Delta}_{i,s}}^{q^{h_2-1}} \\
\hline
\alpha_{s,\Delta_{i,s}}^{q^{h_2}} & \alpha_{s,\bar{\Delta}_{i,s}}^{q^{h_2}} \\
\end{array} \right]
\end{equation*}

Following the proof of Theorem \ref{thm:necsuf} and performing row reduction of $\delta$ coordinates, the resultant matrix  is
\[
\begin{bmatrix}
	\psi_{i, 1} & \psi_{i, 2} & \ldots & \psi_{i, r_1+h_2} \\
	\psi_{i, 1}^{q} & \psi_{i, 2}^{q} & \ldots & \psi_{i, r_1+h_2}^{q} \\
	\vdots & \vdots & \ldots & \vdots \\
	\psi_{i, 1}^{q^{h_2-1}} &  \psi_{i, 2}^{q^{h_2-1}} & \ldots &  \psi_{i, r_1+h_2}^{q^{h_2-1}} \\
	\psi_{i, 1}^{q^{h_2}} &  \psi_{i, 2}^{q^{h_2}} & \ldots &  \psi_{i, r_1+h_2}^{q^{h_2}} \\
\end{bmatrix}
\]	
Now, by Lemma \ref{lem:mds_linearized}, it is the generator matrix of an MDS code if and only if $\Psi_i$ is $(h_2+1)$-wise independent over $\mathbb{F}_q$.
\end{proof}

\subsection{Construction of HL MRC with $h_1=1$ and $h_2=1$}
Now we will describe the construction for the case when there is one mid-level parity per mid-level code ($h_2 = 1$) and one global parity ($h_1=1$). This construction is based on the construction of local MRC with $2$ global parities in \cite{gopi2020maximally}.

\begin{constr} \label{constr:h11h21}
We give a construction of the code $\mathcal{C}$, which is specified by the following parity-check matrix $H$:
\[
	H = \begin{bmatrix}
		H_0 &                        \\
		    & H_0 &                  \\
		    &     & \ddots &         \\
		    &     &        & H_0     \\
		H_1 & H_2 & \hdots & H_{t_1} \\
	\end{bmatrix} \ \ \ \	H_0 = \begin{bmatrix}
		M_0 &                        \\
		    & M_0 &                  \\
		    &     & \ddots &         \\
		    &     &        & M_0     \\
		M_1 & M_2 & \hdots & M_{t_2} \\
	\end{bmatrix}
\]
\[
	M_0 = \begin{bmatrix}
		\A_1       & \A_2       & \ldots & \A_{n_2}       \\
		\A_1^2     & \A_2^2     & \ldots & \A_{n_2}^2     \\
		\vdots     & \vdots     & \ddots & \vdots         \\
		\A_1^{\dl} & \A_2^{\dl} & \ldots & \A_{n_2}^{\dl} \\
	\end{bmatrix}
	\ \ \ \ \ \ \
	M_i = \begin{bmatrix}
		\La_i & \La_i & \ldots & \La_i
	\end{bmatrix}
\]
\[
	\begin{split}
		H_i &= \begin{bmatrix}H_{i,1} & H_{i,2} & \ldots & H_{i, t_2}\end{bmatrix} \\
		H_{i,s} &= \begin{bmatrix}\A_1^{\dl+1} & \A_2^{\dl+1} & \ldots & \A_{n_2}^{\dl+1}\end{bmatrix} \\
	\end{split}
\],

where the following conditions are satisfied:
\bit
\item  q is a prime power such that there exists a subgroup $G$ of $\Fq{}^{*}$ of size atleast $n_2$ and with atleast $t_2$ cosets.
\item $\A_1, \A_2, \ldots \A_{n_2} \in G$ and $\A_i \neq \A_j$.
\item $\La_1, \La_2, \ldots, \La_{t_2} \in \Fq{}^{*}$ be elements from distinct cosets of $G$.
\eit
\end{constr}

We make use of the following determinantal identity to show that the matrix formed by the columns of the parity check matrix corresponding to the erased positions are invertible and hence can be recovered.

\blem[\cite{gopi2020maximally}] \label{lem:diag-id}
Let $C_1,\cdots,C_h$ be $a\times (a+1)$ dimensional matrices and $D_1,\cdots,D_h$ be $h \times (a+1)$ dimensional matrices over a field and let $D_i^{(j)}$ be the $j^{th}$ row of $D_i$. Then,
\begin{align*}
	 & \det \left[
		\begin{array}{c|c|c|c}
			C_1    & 0      & \cdots & 0      \\
			\hline
			0      & C_2    & \cdots & 0      \\
			\hline
			\vdots & \vdots & \ddots & \vdots \\
			\hline
			0      & 0      & \cdots & C_h    \\
			\hline
			D_1    & D_2    & \cdots & D_h    \\
		\end{array}
		\right] \\
		& = (-1)^{\frac{ah(h-1)}{2}}
	\det\left[
		\begin{matrix}
			\det\mattwoone{C_1}{D_1^{(1)}} & \cdots & \det\mattwoone{C_h}{D_h^{(1)}} \\
			\vdots                         & \ddots & \vdots                         \\
			\det\mattwoone{C_1}{D_1^{(h)}} & \cdots & \det\mattwoone{C_h}{D_h^{(h)}} \\
		\end{matrix}
		\right].
\end{align*}
\elem

\bthm

The code $C$ given by Construction \ref{constr:h11h21} is a $ [k, r_1, r_2, h_1=1, h_2=1, \delta] $ HL MRC.

\ethm

\bprf To show that the code is a $[k, r_1, r_2, 1, 1, \dl]$ HL MRC, we consider erasure patterns where there are $\delta$ erasures per local code, one erasure per mid-level code and one more erasure anywhere in the global code. We will show that any such erasure pattern is recoverable.

Since there is only one global erasure and it can be in one mid-level code, we consider that the mid-level code which has additional global erasure has index $l$ and for all $j \neq l$, there are no global erasures associated with these mid-level codes.

Correcting each mid-level code will, in the end, correct the original code.

We show how to correct each of these mid-level codes.
\ben
\item For all $j^{th}$ mid-level codes ($j \neq l$), the corresponding erasure pattern is shown. Let the  mid-level code where the erasure occurs be $j'$.

The submatrix $B_j$ of the parity-check matrix which is used to recover the erasures within the $j^{th}$ mid-level code is given by,
\[
	B_j = \begin{bmatrix}
		\A_{j'_1}       & \A_{j'_2}       & \ldots & \A_{j'_\dl}       & \A_{j'_{\dl+1}}       \\
		\A_{j'_1}^{2}   & \A_{j'_2}^{2}   & \ldots & \A_{j'_\dl}^{2}   & \A_{j'_{\dl+1}}^{2}   \\
		\vdots          & \vdots          & \ddots & \vdots            & \vdots                \\
		\A_{j'_1}^{\dl} & \A_{j'_2}^{\dl} & \ldots & \A_{j'_\dl}^{\dl} & \A_{j'_{\dl+1}}^{\dl} \\
		\La_{j'}        & \La_{j'}        & \ldots & \La_{j'}          & \La_{j'}              \\
	\end{bmatrix}
\]
where $\{j'_1, \ldots j'_{\dl+1}\}$ denote the $\dl+1$ erased coordinates in the local group $j'$.
We can clearly see that this matrix is a Vandermonde matrix after scaling and permuting rows. Hence $\det(B_j) \neq 0$.

\item For the $l^{th}$ mid-level code,
we will also involve the global parity. This case can again be divided into two sub cases depending on the local group where the extra erasure happens:
\ben
\item Both the mid-level erasure and the global erasure occur in the same local code, $l'$.

The matrix formed will be,
\[
	B_l = \begin{bmatrix}
		\A_{l'_1}         & \A_{l'_2}         & \ldots & \A_{l'_\dl}         & \A_{l'_{\dl+1}}         & \A_{l'_{\dl+2}}         \\
		\A_{l'_1}^{2}     & \A_{l'_2}^{2}     & \ldots & \A_{l'_\dl}^{2}     & \A_{l'_{\dl+1}}^{2}     & \A_{l'_{\dl+2}}^{2}     \\
		\vdots            & \vdots            & \ddots & \vdots              & \vdots                  & \vdots                  \\
		\A_{l'_1}^{\dl}   & \A_{l'_2}^{\dl}   & \ldots & \A_{l'_\dl}^{\dl}   & \A_{l'_{\dl+1}}^{\dl}   & \A_{l'_{\dl+2}}^{\dl}   \\
		\La_{l'}          & \La_{l'}          & \ldots & \La_{l'}            & \La_{l'}                & \La_{l'}                \\
		\A_{l'_1}^{\dl+1} & \A_{l'_2}^{\dl+1} & \ldots & \A_{l'_\dl}^{\dl+1} & \A_{l'_{\dl+1}}^{\dl+1} & \A_{l'_{\dl+2}}^{\dl+1} \\
	\end{bmatrix}.
\]
This is similar to above where $B_l$ after scaling and permuting the rows is also a Vandermonde matrix. Hence $\det(B_l) \neq 0$.

\item The mid-level and global erasure occur in different local codes. Let those local codes be $l'$ and $l''$. The matrix $B_l$,
\begin{align*}
	& \det(B_l)& = \\ 
	& \det \begin{bmatrix}
		\A_{l'_1}               & \ldots & \A_{l'_{\dl+1}}                                                                        \\
		\A_{l'_1}^{2}           & \ldots & \A_{l'_{\dl+1}}^{2}                                                                    \\
		\vdots                  & \ddots & \vdots                                                                                 \\
		\A_{l'_1}^{\dl}         & \ldots & \A_{l'_{\dl+1}}^{\dl}                                                                  \\
		                        &        &                         & \A_{l''_1}               & \ldots & \A_{l''_{\dl+1}}         \\
		                        &        &                         & \A_{l''_1}^{2}           & \ldots & \A_{l''_{\dl+1}}^{2}     \\
		                        &        &                         & \vdots                   & \ddots & \vdots                   \\
		                        &        &                         & \A_{l''_1}^{\dl}         & \ldots & \A_{l''_{\dl+1}}^{\dl}   \\
		\La_{l'}                & \ldots & \La_{l'}                & \La_{l''}                & \ldots & \La_{l''}                \\
		\A_{l'_{\dl+1}}^{\dl+1} & \ldots & \A_{l_{\dl+1}'}^{\dl+1} & \A_{l_{\dl+1}''}^{\dl+1} & \ldots & \A_{l_{\dl+1}''}^{\dl+1} \\
	\end{bmatrix}
\end{align*}
\begin{align*}
	  & \hspace{-0.6in}= \det \begin{bmatrix}
		\det \begin{pmatrix}
			\A_{l'_1}       & \ldots & \A_{l'_{\dl+1}}       \\
			\A_{l'_1}^{2}   & \ldots & \A_{l'_{\dl+1}}^{2}   \\
			\vdots          & \ddots & \vdots                \\
			\A_{l'_1}^{\dl} & \ldots & \A_{l'_{\dl+1}}^{\dl} \\
			\La_{l'}        & \ldots & \La_{l'}              \\
		\end{pmatrix} &
		\det \begin{pmatrix}
			\A_{l''_1}       & \ldots & \A_{l''_{\dl+1}}       \\
			\A_{l''_1}^{2}   & \ldots & \A_{l''_{\dl+1}}^{2}   \\
			\vdots           & \ddots & \vdots                 \\
			\A_{l''_1}^{\dl} & \ldots & \A_{l''_{\dl+1}}^{\dl} \\
			\La_{l''}        & \ldots & \La_{l''}              \\
		\end{pmatrix}   \\
		\det \begin{pmatrix}
			\A_{l'_1}               & \ldots & \A_{l'_{\dl+1}}         \\
			\A_{l'_1}^{2}           & \ldots & \A_{l'_{\dl+1}}^{2}     \\
			\vdots                  & \ddots & \vdots                  \\
			\A_{l'_1}^{\dl}         & \ldots & \A_{l'_{\dl+1}}^{\dl}   \\
			\A_{l'_{\dl+1}}^{\dl+1} & \ldots & \A_{l'_{\dl+1}}^{\dl+1} \\
		\end{pmatrix} &
		\det \begin{pmatrix}
			\A_{l''_1}               & \ldots & \A_{l''_{\dl+1}}         \\
			\A_{l''_1}^{2}           & \ldots & \A_{l''_{\dl+1}}^{2}     \\
			\vdots                   & \ddots & \vdots                   \\
			\A_{l''_1}^{\dl}         & \ldots & \A_{l''_{\dl+1}}^{\dl}   \\
			\A_{l_{\dl+1}''}^{\dl+1} & \ldots & \A_{l_{\dl+1}''}^{\dl+1} \\
		\end{pmatrix}   \\
	\end{bmatrix} \\
\end{align*}

The above equation implies that $\det(B_l) = 0$ if and only if
\[
	\det \begin{bmatrix}
		\La_{l'}                      & \La_{l''}                      \\
		\prod_{i=1}^{\dl+1} \A_{l'_i} & \prod_{i=1}^{\dl+1} \A_{l_i''} \\
	\end{bmatrix} = 0.
\]
Where we factored out the non-zero Vandermonde determinants from each column. Since $\A_{l'_i}, \A_{l''_i} \in G$ and $\La_{l'}, \La_{l''}$ are in different cosets of $G$, the last determinant cannot be zero.
\een

\een
Hence, we proved that the code can recover from all possible erasure patterns specified by the definition of HL MRC and hence it is a HL MRC with the corresponding parameters.
\eprf

\subsection{Construction for $h_1=2$ and $h_2=1$}\label{sec:2_1}
Now, we will provide the construction of a HL-MRC with $2$ global parities ($h_1=2$) and $1$ mid-level parity per mid-level code ($h_2 = 1$).

\begin{constr} \label{constr:h12h21}
We give a construction of code $\mathcal{C}$, which is specified by the following parity-check matrix $H$:
\[
	H = \begin{bmatrix}
		H_0 &                        \\
		    & H_0 &                  \\
		    &     & \ddots &         \\
		    &     &        & H_0     \\
		H_1 & H_2 & \hdots & H_{t_1} \\
	\end{bmatrix}
	\ \ \ 
	H_0 = \begin{bmatrix}
		M_0 &                        \\
		    & M_0 &                  \\
		    &     & \ddots &         \\
		    &     &        & M_0     \\
		M_1 & M_2 & \hdots & M_{t_2} \\
	\end{bmatrix}
\]
\[
	M_0 = \begin{bmatrix}
		\frac{1}{\A_1-\B_1}     & \frac{1}{\A_2-\B_1}     & \ldots & \frac{1}{\A_{n_2}-\B_1}     \\
		\frac{1}{\A_1-\B_2}     & \frac{1}{\A_2-\B_2}     & \ldots & \frac{1}{\A_{n_2}-\B_2}     \\
		\vdots                  & \vdots                  & \ddots & \vdots                      \\
		\frac{1}{\A_1-\B_{\dl}} & \frac{1}{\A_2-\B_{\dl}} & \ldots & \frac{1}{\A_{n_2}-\B_{\dl}} \\
	\end{bmatrix}
\]
\[	M_i = \begin{bmatrix}
		\frac{1}{\A_1-\B_{\dl+1}} & \frac{1}{\A_2-\B_{\dl+1}} & \ldots & \frac{1}{\A_{n_2}-\B_{\dl+1}} \\
	\end{bmatrix}
\]
\[
	\begin{split}
		H_i &= \begin{bmatrix}H_{i,1} & H_{i,2} & \ldots & H_{i, t_2}\end{bmatrix} \\
		H_{i,s} &=
		\begin{bmatrix}
			\frac{\La_{s+(i-1)t_2}}{\A_1-\B_{\dl+2}} & \frac{\La_{s+(i-1)t_2}}{\A_2-\B_{\dl+2}} & \ldots & \frac{\La_{s+(i-1)t_2}}{\A_{n_2}-\B_{\dl+2}} \\
			\frac{\mu_{s+(i-1)t_2}}{\A_1-\B_{\dl+3}} & \frac{\mu_{s+(i-1)t_2}}{\A_2-\B_{\dl+3}} & \ldots & \frac{\mu_{s+(i-1)t_2}}{\A_{n_2}-\B_{\dl+3}} \\
		\end{bmatrix} \\
	\end{split}
\]

The parameters described in the above parity-check matrix are picked as follows:
\bit
\item $q_0 \geq 2(n_2+\dl)+3$ is a prime power.
\item There exists a subgroup $G$ of $\Fqo^{*}$ of size at least $n_2+2$ with atleast $t_1t_2$ cosets.
\item $\Fq{}$ is an extension field of $\Fqo$.
\item $\mu_1, \ldots, \mu_{t_1t_2}$ are picked from distinct cosets of $G$.
\item Choose distinct $\B_{\dl+1}, \B_{\dl+2}, \B_{\dl+3} \in \Fqo$.
\item Pick $\A_1, \ldots \A_{n_2} \in \Fqo$ such that, $\frac{\A_i-\B_{\dl+2}}{\A_i - \B_{\dl+3}}, \frac{\A_i - \B_{\dl+1}}{\A_i - \B_{\dl+3}} \in G$.
\item Pick distinct $\B_1, \ldots, \B_{\dl} \in \Fqo \setminus \{\A_1, \ldots, \A_{n_2}, \B_{\dl+1}, \B_{\dl+2}, \B_{\dl+3}\}$.
\item $\La_1, \La_2, \ldots, \La_{t_1t_2} \in \Fq{}$ are picked 4 wise-independent over $\Fqo$.
\eit
\end{constr}

\begin{thm} \label{thm:h12h21}
The code $C$ given by Construction \ref{constr:h12h21} is a $ [k, r_1, r_2, h_1=1, h_2=1, \delta] $ HL-MRC.
\end{thm}

\bprf

Again as in previous proof, we consider the case when there are $\delta$ erasures per mid-level code, one erasure per mid-level code and two more global erasures anywhere in the code. We again look at the erasure patterns within each mid-level codes. The following distinct patterns are possible with respect to the mid-level codes.
\ben
\item No global erasures occur in that mid-level code.
\item Either one or both of the global erasures occur in the mid-level code.
\een
We prove that all the above erasure patterns are recoverable.

\ben

\item When no global erasures occur in the mid-level code, there are $\dl$ erasures per local code and one more erasure per mid-level code.

In this scenario, we involve the mid-level parities. Let $l$ be the affected mid-level code and $l'$ be the local code within the mid-level code where the erasure occurs. 
Let $\g_{i,j} = \frac{1}{\A_j-\B_i}$.
The matrix, $B_l$
\[
	B_l = \begin{bmatrix}
		\g_{1,l'_1}     & \g_{1,l'_2}     & \ldots & \g_{1,l'_{\dl+1}}      \\
		\g_{2,l'_1}     & \g_{2,l'_2}     & \ldots & \g_{2,l'_{\dl+1}}      \\
		\vdots          & \vdots          & \ddots & \vdots                 \\
		\g_{\dl,l'_1}   & \g_{\dl,l'_2}   & \ldots & \g_{\dl,l'_{\dl+1}}    \\
		\g_{\dl+1,l'_1} & \g_{\dl+1,l'_2} & \ldots & \g_{\dl+1, l'_{\dl+1}} \\
	\end{bmatrix}.
\]
Where $\{l'_1, l'_2, \ldots, l'_{\dl+1} \}$ are the erased coordinates in local code $l'$. This is a Cauchy matrix and hence $\det(B_l) \neq 0$.

\item When there are global erasures, there are $\dl$ erasures per local code, one erasure per mid-level code and two more erasures anywhere in the code

We will only list the cases (6 in total) of erasure patterns here. We refer the reader to Appendix A for details of the proof, where we derive that in each of the following cases, the parity-check matrix restricted to the  erased columns is full rank.
\ben
\item Both global erasures are in the same local code as the mid-level code.
\item Both global erasures are in the same local code but different one from the mid-level erasure for that mid-level code.
\item Both global and the one mid-level erasures are in different local code but the same mid-level code.
\item Both global erasures are in different mid-level code but share that local code with the mid-level parities for that mid-level code.
\item Each global erasure is in their own different local code and do not share with the mid-level erasures.
\item In this case, one of the global erasure shares the local code with a mid-level code while the other does not.
\een
\een

\eprf

\section{A Field Size Lower Bound for HL MRC}

In this section, we will derive lower bounds on the field size of HL MRC. The proof technique is similar to the one developed in \cite{gopi2020maximally}, with the difference being that in this case, there are mid-level codes as well and hence while performing shortening in the parity check matrix, this has to be taken into account. The following lemma derived in \cite{gopi2020maximally} will be useful in deriving the lower bounds on field size.

\blem[\cite{gopi2020maximally}] \label{lem:projective}
Let $X_1, X_2, \ldots, X_g \in \mathbb{P}^d(\mathbb{F}_q)$ be mutually disjoint subsets each of size $t$ with $g \geq d+1$ of the projective space $\mathbb{P}^d(\mathbb{F}_q)$ . If $q < \left ( \frac{g}{d}-1 \right )t - 4$. Then, there exists a hyperplane which intersects $d+1$ distinct subsets among $X_1, X_2, \ldots, X_g$.
\elem

\bthm \label{thm:acase}

Consider a $[k,r_1,r_2,h_1,h_2,\delta]$ HL MRC. If $(\delta+2) \leq h_1 + h_2$, $h_1 \leq \frac{n}{n_1}$ and $h_2 \leq \frac{n_1}{n_2}-1$, then the field size $q$ is lower bounded as follows:
\begin{equation}
q \geq \left ( \frac{\frac{n}{n_2}}{h_1h_2+h_1-1} - 1 \right ) {r+\delta \choose \delta+1} - 4.
\end{equation}
\ethm

\begin{proof}
Consider an arbitrary $[k,r_1,r_2,h_1,h_2,\delta]$ HL MRC with $t_1 = \frac{n}{n_1}$ mid-level codes and $t_2 = \frac{n_1}{n_2}$ local codes per mid-level code. The code has parity-check matrix of the form
\[
	H = \begin{bmatrix}
		M_1 &                        \\
		    & M_2 &                  \\
		    &     & \ddots &         \\
		    &     &        & M_{t_1}     \\
		P_1 & P_2 & \hdots & P_{t_1} \\
	\end{bmatrix}, \ \ M_i = \begin{bmatrix}
		M_{i,1} &                        \\
		    & M_{i,2} &                  \\
		    &     & \ddots &         \\
		    &     &        & M_{i,t_2}     \\
		N_{i,1} & N_{i,2} & \hdots & N_{i,t_2} \\
	\end{bmatrix},
\]

where $M_1, \ldots, M_{t_1}$ are $\delta' \times (r_1 + \delta')$ matrices over $\mathbb{F}_q$, where $\delta' = h_2 + \left ( \frac{r_1+h_2}{r_2} \right ) \delta$. $P_1, \ldots, P_{t_1}$ are $h_1 \times r_1 + \delta'$ matrices over $\mathbb{F}_q$. $M_{i,1}, \ldots, M_{i, t_2}$ are $\delta \times (r_2 + \delta)$ matrices over $\mathbb{F}_q$. $N_{i,1}, \ldots, N_{i, t_2}$ are $h_2 \times r_2 + \delta$ matrices over $\mathbb{F}_q$. For every subset $S \subseteq [r+\delta]$ of size $|S| = \delta +1$, $M_{i,j}(S)$ is an $\delta \times (\delta +1)$ matrix is full rank. Let $M_{i,j}(S)^\perp \in \mathbb{F}_q^{\delta+1}$ be a nonzero vector orthogonal to the row space of $M_{i,j}(S)$. We know that $M_{i,j}(S) M_{i,j}(S)^\perp = 0$, $q^2_{i,j}(S) = N_{i,j}(S) M_{i,j}(S)^\perp$ where $q^2_{i,j}(S)$ is a $h_2 \times 1$ matrix and $q^1_{i,j}(S) = P_{i,j}(S) M_{i,j}(S)^\perp$ where $q^1_{i,j}(S)$ is a $h_1 \times 1$ matrix. $M_i(S_i)^\perp$ is defined as follows:
\[	M_i(S_i)^\perp = \begin{bmatrix}
		M_{i,1}(S_{i,1})^\perp &                        \\
		    & M_{i,2}(S_{i,2})^\perp &                  \\
		    &     & \ddots &         \\
		    &     &        & M_{i,t_2}(S_{i,t_2})^\perp     \\
	\end{bmatrix}.
\]
Let
\[
D = \begin{bmatrix}
		M_1(S_1) &                        \\
		    & M_2(S_2) &                  \\
		    &     & \ddots &         \\
		    &     &        & M_{t_1}(S_{t_1})     \\
		P_1(S_1) & P_2(S_2) & \hdots & P_{t_1}(S_{t_1}) \\
	\end{bmatrix}.
\]
We form the matrix 
\[ Q = D \ diag(M_1(S_1)^\perp, M_2(S_2)^\perp, \ldots, M_{t_1}(S_{t_1})^\perp).
\]
After removing the zero rows in the $Q$ matrix, the structure of the resulting matrix is as follows:
\[
Q' = \begin{bmatrix}
		q^2_{l_{1,1}, S_{1,1}} & \ldots & q^2_{l_{1,h_3}, S_{1,h_3}} & & & \\
		    &     & \ddots &         \\
		    &     &        &      \\
		& &  & \ldots & q^2_{l_{h_1,1}, S_{h_1,1}} & \ldots & q^2_{l_{h_1,h_3}, S_{h_1,h_3}} \\
		q^1_{l_{1,1}, S_{1,1}} & \ldots & q^1_{l_{1,h_3}, S_{1,h_3}} & \ldots & q^1_{l_{h_1,1}, S_{h_1,1}} & \ldots & q^1_{l_{h_1,h_3}, S_{h_1,h_3}}
	\end{bmatrix}.
\]
We denote the first block of columns of the above matrix by $q_{l_{1,1}, S_{1,1}}$, 
second by $q_{l_{1,2}, S_{1,2}}$ and the last one by $q_{l_{h_1,h_2+1}, S_{h_1,h_2+1}}$. With this notation, the following lemma gives the rank property of the above matrix. 

\begin{lem} \label{lem:lem1}
For any $I \subset [t_1]$, with $|I| = h_1$ and for every $i \in I$, $l_{i,1}, l_{i,2}, \ldots, l_{i, (h_2+1)}$ and subsets $S_{i,1}, S_{i,2}, \ldots, S_{i, h_2+1} \subseteq [r+\delta]$ of size $\delta + 1 $ each, Then $(h_1h_2 + h_1 ) \times (h_1 h_2 + h_1)$ matrix $Q'$ is full rank.
\end{lem}
\bprf
Follows from the HL MRC property of the code under consideration.
\eprf

\begin{lem} \label{lem:lem2}
For every $i,j \in [t_1,t_2]$, no two vectors in $\{ q_{i,j}(S) : S \subseteq {[r+\delta] \choose \delta+1} \}$ are multiples of each other.
\end{lem}
\begin{proof}
Suppose $q_{i,j}(S) = \lambda q_{i,j}(T)$ for some distinct $S, T \subset [r+\delta]$ of size $\delta +1$ each and some nonzero $\lambda \in \mathbb{F}_q$.
\begin{eqnarray*}
& & \begin{bmatrix} M_{i,j}(S) \\ N_{i,j}(S) \\ P_{i,j}(S)\end{bmatrix} M_{i,j}(S)^\perp - \lambda \begin{bmatrix} M_{i,j}(T) \\ N_{i,j}(T) \\ P_{i,j}(T)\end{bmatrix}
M_{i,j}(T)^\perp \\
& & \begin{bmatrix} 0 \\ q^2_{i,j}(S) \\ q^1_{i,j}(S) \end{bmatrix} - \lambda \begin{bmatrix} 0 \\ q^2_{i,j}(T) \\ q^1_{i,j}(T) \end{bmatrix} = 0.
\end{eqnarray*}
Note that every coordinate of $M_{i,j}(S)^\perp$ is nonzero. Otherwise, it implies a linear dependence among $\delta$ columns of $M_{i,j}(S)$. Thus, we have a linear combination of $\begin{bmatrix} M_{i,j}(S \cup T) \\ N_{i,j}(S \cup T) \\ P_{i,j}(S \cup T)\end{bmatrix}$. However, $|S \cup T| \leq 2\delta+2 \leq \delta + h_1+h_2$. By the MR property, any set of columns of the matrix $\begin{bmatrix} M_{i,j} \\ N_{i,j} \\ P_{i,j}\end{bmatrix}$ of size $\delta + h_1+h_2$ has to be full rank. Thus, we arrive at a contradiction. Thus, no two vectors in $\{ q_{i,j}(S) : S \subseteq {[r+\delta] \choose \delta+1} \}$ are multiples of each other.
\end{proof}
By the above lemma, we can think of $\{ q_{i,j}(S) : S \subseteq {[r+\delta] \choose \delta+1} \}$ as distinct points in $\mathbb{P}^{(h_1h_2+h_1-1)}(\mathbb{F}_q)$. Define set $X_{i,j}$ as follows: $X_{i,j} = \{ q_{i,j}(S) : S \subseteq {[r+\delta] \choose \delta+1} \}$.
The sets $\{X_{i,j}, i \in [t_1], j \in [t_2]\}$ are all mutually disjoint. Since $t_1 \geq h_1$ and $t_2 \geq h_2 +1$, it follows that $t_1 t_2 \geq h_1h_2 + h_1$.
Based on Lemma \ref{lem:lem2}, there is no hyperplane in $\mathbb{P}^{(h_1h_2+h_1-1)}(\mathbb{F}_q)$ which contains $h_1h_2 +h_1$ points from distinct subsets of $\{X_{i,j}, i \in [t_1], j \in [t_2]\}$. By Lemma \ref{lem:projective}, we have lower bound on the field size.

\end{proof}

\bthm

Consider a $[k,r_1,r_2,h_1,h_2,\delta]$ HL MRC. If $4 \leq h_1+h_2 \leq (\delta+2)$, $h_1 \leq \frac{n}{n_1}$ and $h_2 \leq \frac{n_1}{n_2}-1$, then the field size $q$ is lower bounded as follows:
\begin{equation}
q \geq \left ( \frac{\frac{n}{n_2}}{h_1h_2+h_1-1} - 1 \right ) {r+h_1+h_2-2 \choose h_1+h_2-1} - 4.
\end{equation}
\ethm

\begin{proof}
In this case, we do not take arbitrary $S$ and $T$ as in the case of proof of Theorem \ref{thm:acase} but consider subsets $S$ that have size $\delta+1$ but constrained to contain the subset $\{1,2,\ldots,(\delta+2-h_1-h_2)$. By picking the sets in this way, we still ensure that the pairwise unions have size atmost $\delta+h_1+h_2$. The total number of such sets is given by ${r+h_1+h_2-2 \choose h_1+h_2-1}$. Based on this counting, the statement of the theorem follows.
\end{proof}

\bthm

Consider a $[k,r_1,r_2,h_1,h_2,\delta]$ HL-MRC. If $(\delta+2) \leq h_1 + h_2$, $h_1 > \frac{n}{n_1}$ and $h_2 \leq \frac{n_1}{n_2}- \lceil \frac{h_1}{t_1} \rceil$, then the field size $q$ is lower bounded as follows:
\begin{equation}
q \geq \left ( \frac{\frac{n}{n_2}}{\frac{n}{n_1}h_2+h_1-1} - 1 \right ) {r+\delta \choose \delta+1} - 4.
\end{equation}
\ethm

\begin{proof}
Let $f_1 \geq f_2 \geq \ldots \geq f_{t_1}$ be such that $f_i = \lceil \frac{h_1}{t_1} \rceil$ or $\lfloor \frac{h_1}{t_1} \rfloor$ and $\sum_{i=1}^{t_1} f_i = h_1$. From each $i^{th}$ middle code, pick $h_2 + f_i$ local codes and $\delta+1$ columns from every local code. By applying the row reduction similar to the proof of Theorem \ref{thm:acase} and applying appropriately modifying versions of Lemmas \ref{lem:lem1} and \ref{lem:lem2}, the result follows.
\end{proof}

\begin{note}
Please note that we derive the field size bounds for the cases when (i) $(\delta+2) \leq h_1 + h_2$, $h_1 \leq \frac{n}{n_1}$ and $h_2 \leq \frac{n_1}{n_2}-1$ (ii) $(\delta+2) \leq h_1 + h_2$, $h_1 > \frac{n}{n_1}$ and $h_2 \leq \frac{n_1}{n_2}- \lceil \frac{h_1}{t_1} \rceil$. There are other cases of parameters $\delta, h_1, h_2$ for which field-size bounds need to be derived and we leave it as part of future work.
\end{note}

\appendices

\section{Proof of Theorem \ref{thm:h12h21}} \label{app:proof}

The following results related to the determinants of matrices will be useful in proving Theorem \ref{thm:h12h21}.

\blem[\cite{gopi2020maximally}] \label{lem:2prod-id}
Let $C_1$ be an $a\times (a+1)$ matrix, $C_2$ be an $a\times (a+2)$ matrix, $D_1$ be a $3 \times (a+1)$ matrix and $D_2$ be a $3\times (a+2)$ matrix and let  $D_i^{(j)}$ be the $j^{th}$ row of $D_i$. Then,
\begin{align*}
	\det \left[ \begin{array}{c|c}
			C_1 & 0   \\
			\hline
			0   & C_2 \\
			\hline
			D_1 & D_2 \\
		\end{array}
		\right]
	= &
	(-1)^a \cdot \Bigg(
	\det\mattwoone{C_1}{D_1^{(1)}} \cdot \det\matthreeone{C_2}{D_2^{(2)}}{D_2^{(3)}}
	- \det\mattwoone{C_1}{D_1^{(2)}} \cdot \det\matthreeone{C_2}{D_2^{(1)}}{D_2^{(3)}}                                \\
	  & \phantom{(-1)^a \cdot \Bigg(\det\mattwoone{C_1}{D_1^{(1)}} \cdot \det\matthreeone{C_2}{D_2^{(2)}}{D_2^{(3)}}} 
	+ \det\mattwoone{C_1}{D_1^{(3)}}\cdot \det\matthreeone{C_2}{D_2^{(1)}}{D_2^{(2)}}
	\Bigg)
\end{align*}
\elem

\blem[\cite{gopi2020maximally}] \label{lem:3prod-id}
Given $C_1$ and $C_2$ to be $a\times a+1$ matrices and $C_3$ to be an $a \times (a+2)$ matrix. Also, $D_1$ and $D_2$ are $4 \times (a+1)$ matrices while $D_3$ is a $4 \times (a+2)$ matrix. It is also given that $D_3^{(1)}, D_1^{(2)}, D_2^{(2)} = [0]$.
Then,
\begin{align*}
	\det \left[
		\begin{array}{c|c|c}
			C_1 & 0   & 0   \\
			\hline
			0   & C_2 & 0   \\
			\hline
			0   & 0   & C_3 \\
			\hline
			D_1 & D_2 & D_3 \\
		\end{array}
	\right] = (-1)^a \cdot \Bigg( &
	\det \begin{pmatrix}C_{1} \\ D_{1}^{(1)} \end{pmatrix} \cdot
	\det \begin{pmatrix}C_{2} \\ D_{2}^{(3)} \end{pmatrix} \cdot
	\det \begin{pmatrix}C_{3} \\ D_{3}^{(2)} \\ D_{3}^{(4)}\end{pmatrix}                                         \\
	                              & + \det \begin{pmatrix}C_{1} \\ D_{1}^{(1)} \end{pmatrix} \cdot
	\det \begin{pmatrix}C_{2} \\ D_{2}^{(4)} \end{pmatrix} \cdot
	\det \begin{pmatrix}C_{3} \\ D_{3}^{(2)} \\ D_{3}^{(3)}\end{pmatrix}                                         \\
	                              & +
	\det \begin{pmatrix}C_{1} \\ D_{1}^{(3)} \end{pmatrix} \cdot
	\det \begin{pmatrix}C_{2} \\ D_{2}^{(1)} \end{pmatrix} \cdot
	\det \begin{pmatrix}C_{3} \\ D_{3}^{(2)} \\ D_{3}^{(4)}\end{pmatrix}                                         \\
	                              & - \det \begin{pmatrix}C_{1} \\ D_{1}^{(4)} \end{pmatrix} \cdot
	\det \begin{pmatrix}C_{2} \\ D_{2}^{(1)} \end{pmatrix} \cdot
	\det \begin{pmatrix}C_{3} \\ D_{3}^{(2)} \\ D_{3}^{(3)}\end{pmatrix}
	\Bigg)
\end{align*}
\elem

\bprf
Follow as a result of Lemmas B.2 in \cite{gopi2020maximally}.
\eprf

We also define a cauchy matrix here.
\blem[Cauchy Matrix \cite{roth2006}]
Let $a_1, a_2, \ldots, a_n, b_1, b_2, \ldots, b_n \in \Fq{}$ be all distinct. Then,
\[
	\det \begin{bmatrix}
		\frac{1}{a_1-b_1} & \frac{1}{a_2-b_1} & \ldots & \frac{1}{a_n-b_1} \\
		\frac{1}{a_1-b_2} & \frac{1}{a_2-b_2} & \ldots & \frac{1}{a_n-b_2} \\
		\vdots            & \vdots            & \ddots & \vdots            \\
		\frac{1}{a_1-b_n} & \frac{1}{a_2-b_n} & \ldots & \frac{1}{a_n-b_n} \\
	\end{bmatrix} =
	\frac{\prod_{i>j}(a_i-a_j)(b_i-b_j)}{\prod_{i,j}(a_i-b_j)}.
\]
Such a matrix is called an Cauchy Matrix. Every minor of a Cauchy matrix is also an Cauchy matrix.
\elem

Again as in previous proof, we consider the case when there are $\delta$ erasures per local code, one erasure per mid-level code and two more global erasures anywhere in the code. We again look at the erasure patterns within each mid-level codes. There are three distinct patterns possible
\ben
\item No global erasures occur in that mid-level code.
\item Either one or both of the global erasures occur in the mid-level code.
\een
We that each of the above are correctible.

Let $\g_{i,j} = \frac{1}{\A_j-\B_i}$.
\ben

\item When no global erasures occur in the mid-level code, there are $\dl$ erasures per local code and one more erasure per mid-level code.

In this scenario, we involve the mid-level parities. Let $l$ be the affected mid-level code and $l'$ be the local code within the mid-level code where the erasure occurs. The matrix, $B_l$
\[
	B_l = \begin{bmatrix}
		\g_{1,l'_1}     & \g_{1,l'_2}     & \ldots & \g_{1,l'_{\dl+1}}      \\
		\g_{2,l'_1}     & \g_{2,l'_2}     & \ldots & \g_{2,l'_{\dl+1}}      \\
		\vdots          & \vdots          & \ddots & \vdots                 \\
		\g_{\dl,l'_1}   & \g_{\dl,l'_2}   & \ldots & \g_{\dl,l'_{\dl+1}}    \\
		\g_{\dl+1,l'_1} & \g_{\dl+1,l'_2} & \ldots & \g_{\dl+1, l'_{\dl+1}} \\
	\end{bmatrix}.
\]
Where $\{l'_1, l'_2, \ldots, l'_{\dl+1} \}$ are the erased coordinates in local code $l'$. This is a Cauchy matrix and hence $\det(B_l) \neq 0$.

\item When there are global erasures, there are $\dl$ erasures per local code, one erasure per mid-level code and two more erasures anywhere in the code

Here we have a lot more sub-cases.
\ben
\item Both global erasures are in the same local code as the mid-level code. Let $l$ be the affected mid-level code and $l'$ be the local code in the mid-level code where the erasure happens.
The matrix $B_l$ in that case,
\[
	B_l = \begin{bmatrix}
		\cauchy{l'}{\dl+3}{\dl+1}          \\
		\cauchyrowm{l'}{\dl+3}{\dl+2}{\La} \\
		\cauchyrowm{l'}{\dl+3}{\dl+3}{\mu}
	\end{bmatrix}.
\]
This is also a Cauchy matrix with the last two rows scaled to $\La_{l'}$ and $\mu_{l'}$ respectively. Hence $\det(B_l) \neq 0$ and this erasure pattern is correctible.

\item Both global erasures are in the same local code but different one from the mid-level erasure for that mid-level code.

Assume that the $l^{th}$ mid-level code is affected. Let $l''$ be the local code with two erasures while $l'$ be the other one within this mid-level code.
\begin{align*}
	B_l=\begin{bmatrix}
		\begin{matrix}\cauchy{l'}{\dl+1}{\dl}\end{matrix} &                            \\
		                           & \begin{matrix}\cauchy{l''}{\dl+2}{\dl}\end{matrix} \\
		\begin{matrix}\cauchyrowp{l'}{\dl+1}{\dl+1}\end{matrix} & \begin{matrix}\cauchyrowp{l''}{\dl+2}{\dl+1}\end{matrix} \\
		\begin{matrix}\cauchyrowm{l'}{\dl+1}{\dl+2}{\La}\end{matrix} & \begin{matrix}\cauchyrowm{l''}{\dl+2}{\dl+2}{\La}\end{matrix} \\
		\begin{matrix}\cauchyrowm{l'}{\dl+1}{\dl+3}{\mu}\end{matrix} & \begin{matrix}\cauchyrowm{l''}{\dl+2}{\dl+3}{\mu}\end{matrix} \\
	\end{bmatrix}.
\end{align*}
Expanding this via the lemma \ref{lem:2prod-id},
\begin{align*}
	\hspace{-5em}
	\det(B_l) = & \det\bpm
	\cauchy{l'}{\dl+1}{\dl}             \\
	\cauchyrowp{l'}{\dl+1}{\dl+1}
	\epm \cdot
	\det\bpm
	\cauchy{l''}{\dl+2}{\dl}            \\
	\cauchyrowm{l''}{\dl+2}{\dl+2}{\La} \\
	\cauchyrowm{l''}{\dl+2}{\dl+3}{\mu} \\
	\epm                                \\
	            & - \det\bpm
	\cauchy{l'}{\dl+1}{\dl}             \\
	\cauchyrowm{l'}{\dl+1}{\dl+2}{\La}  \\
	\epm \cdot
	\det\bpm
	\cauchy{l''}{\dl+2}{\dl}            \\
	\cauchyrowp{l''}{\dl+2}{\dl+1}      \\
	\cauchyrowm{l''}{\dl+2}{\dl+3}{\mu} \\
	\epm                                \\
	            & + \det\bpm
	\cauchy{l'}{\dl+1}{\dl}             \\
	\cauchyrowm{l'}{\dl+1}{\dl+3}{\mu}  \\
	\epm \cdot
	\det\bpm
	\cauchy{l''}{\dl+2}{\dl}            \\
	\cauchyrowp{l''}{\dl+2}{\dl+1}      \\
	\cauchyrowm{l''}{\dl+2}{\dl+2}{\La} \\
	\epm                                \\
\end{align*}

\begin{align*}
	\hspace{-5em}
	\det(B_l) = & \La_{l''}\mu_{l''} \cdot \det\bpm
	\cauchy{l'}{\dl+1}{\dl}                          \\
	\cauchyrowp{l'}{\dl+1}{\dl+1}
	\epm \cdot
	\det\bpm
	\cauchy{l''}{\dl+2}{\dl}                         \\
	\cauchyrowp{l''}{\dl+2}{\dl+2}                   \\
	\cauchyrowp{l''}{\dl+2}{\dl+3}                   \\
	\epm                                             \\
	            & - \La_{l'}\mu_{l''} \cdot \det\bpm
	\cauchy{l'}{\dl+1}{\dl}                          \\
	\cauchyrowp{l'}{\dl+1}{\dl+2}                    \\
	\epm \cdot
	\det\bpm
	\cauchy{l''}{\dl+2}{\dl}                         \\
	\cauchyrowp{l''}{\dl+2}{\dl+1}                   \\
	\cauchyrowp{l''}{\dl+2}{\dl+3}                   \\
	\epm                                             \\
	            & + \La_{l''}\mu_{l'} \cdot \det\bpm
	\cauchy{l'}{\dl+1}{\dl}                          \\
	\cauchyrowp{l'}{\dl+1}{\dl+3}                    \\
	\epm \cdot
	\det\bpm
	\cauchy{l''}{\dl+2}{\dl}                         \\
	\cauchyrowp{l''}{\dl+2}{\dl+1}                   \\
	\cauchyrowp{l''}{\dl+2}{\dl+2}                   \\
	\epm                                             \\
\end{align*}

Each term in this determinant is $\La_i\mu_j$ multiplied by a Cauchy matrix $\in \Fqo$. The determinant is again a linear combination of $\La_{l'}$ and $\La_{l''}$. Again, this determinant cannot be zero because $\La$'s are 4-wise independent.

\item Both global and the one mid-level erasures are in different local code but the same mid-level code.

Let the affected mid-level code be $l$ and the local codes within, where the erasure occurs, be $l^{(1)}, l^{(2)}$ and $l^{(3)}$. The matrix $B_l$,
\[ \hspace{-5em}
	\scalemath{0.9}{
		B_l = \begin{bmatrix}
			\begin{matrix}\cauchy{l^{(1)}}{\dl+1}{\dl} \end{matrix}                                                           \\

			                           & \begin{matrix}\cauchy{l^{(2)}}{\dl+1}{\dl}\end{matrix}                              \\
			                           &                            & \begin{matrix}\cauchy{l^{(3)}}{\dl+1}{\dl}\end{matrix} \\

			\begin{matrix}\cauchyrowp{l^{(1)}}{\dl+1}{\dl+1}\end{matrix} & \begin{matrix}\cauchyrowp{l^{(2)}}{\dl+1}{\dl+1}\end{matrix} & \begin{matrix}\cauchyrowp{l^{(3)}}{\dl+1}{\dl+1}\end{matrix} \\
			\begin{matrix}\cauchyrowm{l^{(1)}}{\dl+1}{\dl+2}{\La}\end{matrix} & \begin{matrix}\cauchyrowm{l^{(2)}}{\dl+1}{\dl+2}{\La}\end{matrix} & \begin{matrix}\cauchyrowm{l^{(3)}}{\dl+1}{\dl+2}{\La}\end{matrix} \\

			\begin{matrix}\cauchyrowm{l^{(1)}}{\dl+1}{\dl+3}{\mu}\end{matrix} & \begin{matrix}\cauchyrowm{l^{(2)}}{\dl+1}{\dl+3}{\mu}\end{matrix} & \begin{matrix}\cauchyrowm{l^{(3)}}{\dl+1}{\dl+3}{\mu}\end{matrix} \\
		\end{bmatrix}.
	}\]

$\det(B_l)$ can be expanded via lemma \ref{lem:diag-id}. After doing that and setting the determinant to zero,
\[
	\det(B_l) = 0,
\]
we get,
\begin{align*}
	\hspace{-8em}
	                                                                                                                                   & \det \bbm
	\frac{c_{l^{(1)}}d\prod_{i \in [\dl]}(\B_i-\B_{\dl+1})}{e_{l^{(1)}} \prod_{i \in l^{(1)}_S} (\A_i-\B_{\dl+1})}                     &
	\frac{c_{l^{(2)}}d\prod_{i \in [\dl]}(\B_i-\B_{\dl+1})}{e_{l^{(2)}} \prod_{i \in l^{(2)}_S} (\A_i-\B_{\dl+1})}                     &
	\frac{c_{l^{(3)}}d\prod_{i \in [\dl]}(\B_i-\B_{\dl+1})}{e_{l^{(3)}} \prod_{i \in l^{(3)}_S} (\A_i-\B_{\dl+1})}                                                                                                                                                                                     \\
	\La_{l^{(1)}} \cdot \frac{c_{l^{(1)}}d\prod_{i \in [\dl]}(\B_i-\B_{\dl+2})}{e_{l^{(1)}} \prod_{i \in l^{(1)}_S} (\A_i-\B_{\dl+2})} &
	\La_{l^{(2)}} \cdot \frac{c_{l^{(2)}}d\prod_{i \in [\dl]}(\B_i-\B_{\dl+2})}{e_{l^{(2)}} \prod_{i \in l^{(2)}_S} (\A_i-\B_{\dl+2})} &
	\La_{l^{(3)}} \cdot \frac{c_{l^{(3)}}d\prod_{i \in [\dl]}(\B_i-\B_{\dl+2})}{e_{l^{(3)}} \prod_{i \in l^{(3)}_S} (\A_i-\B_{\dl+2})}                                                                                                                                                                 \\
	\mu_{l^{(1)}} \cdot \frac{c_{l^{(1)}}d\prod_{i \in [\dl]}(\B_i-\B_{\dl+3})}{e_{l^{(1)}} \prod_{i \in l^{(1)}_S} (\A_i-\B_{\dl+3})} &
	\mu_{l^{(2)}} \cdot \frac{c_{l^{(2)}}d\prod_{i \in [\dl]}(\B_i-\B_{\dl+3})}{e_{l^{(2)}} \prod_{i \in l^{(2)}_S} (\A_i-\B_{\dl+3})} &
	\mu_{l^{(3)}} \cdot \frac{c_{l^{(3)}}d\prod_{i \in [\dl]}(\B_i-\B_{\dl+3})}{e_{l^{(3)}} \prod_{i \in l^{(3)}_S} (\A_i-\B_{\dl+3})}                                                                                                                                                                 \\
	\ebm                                                                                                                               & = 0                                                                                                                                                           \\
	                                                                                                                                   & \det\bbm
	1                                                                                                                                  & 1                                                                             & 1                                                                             \\
	\La_{l^{(1)}} \prod_{i \in l^{(1)}_S} \frac{\A_i-\B_{\dl+1}}{\A_i-\B_{\dl+2}}                                                      & \La_{l^{(2)}} \prod_{i \in l^{(2)}_S} \frac{\A_i-\B_{\dl+1}}{\A_i-\B_{\dl+2}} & \La_{l^{(3)}} \prod_{i \in l^{(3)}_S} \frac{\A_i-\B_{\dl+1}}{\A_i-\B_{\dl+3}} \\
	\mu_{l^{(1)}} \prod_{i \in l^{(1)}_S} \frac{\A_i-\B_{\dl+1}}{\A_i-\B_{\dl+3}}                                                      & \mu_{l^{(2)}} \prod_{i \in l^{(2)}_S} \frac{\A_i-\B_{\dl+1}}{\A_i-\B_{\dl+2}} & \mu_{l^{(3)}} \prod_{i \in l^{(3)}_S} \frac{\A_i-\B_{\dl+1}}{\A_i-\B_{\dl+3}} \\
	\ebm                                                                                                                               & = 0
\end{align*}
Where, \bit
\item $l^{(i)}_S = \{l^{(i)}_{1}, \ldots, l^{(i)}_{\dl+1}\}$.
\item $c_{l^{(i)}} = \prod^{}_{f>g, f,g \in  l^{(i)}_S} (\A_f - \A_g)$.
\item $d = \prod_{f>g, f,g \in [\dl]}^{} (\B_f - \B_g)$.
\item $e_{l^{(i)}} = \prod_{f \in l^{(i)}_S, g \in [\dl]}^{} (\A_f - \B_g)$.
\eit
Now, by the choice of $\A$'s, $\prod_{i \in l^{(k)}_{S}} \frac{\A_i-\B_{\dl+1}}{\A_i-\B_{\dl+3}} \in G $. And because $\mu_i$ belong to different cosets in $G$, the last row in the above matrix consists of distinct elements. This determinant is a linear combination in the three $\La$'s. Hence the determinant is non-zero because the $\La$'s are 4-wise independent.

\item Both global erasures are in different mid-level code but share that local code with the mid-level parities for that mid-level code.

Assume $k^{th}$ and $l^{th}$ mid-level codes are affected. The local codes within them, where the erasure occurs, are $k'$ and $l'$. The matrix $B_{k,l}$,
\[
	B_{k,l} = \begin{bmatrix}
		\begin{matrix}\cauchy{k'}{\dl+2}{\dl+1}\end{matrix} &                            \\
		                           & \begin{matrix}\cauchy{l'}{\dl+2}{\dl+1}\end{matrix} \\
		\begin{matrix}\cauchyrowm{k'}{\dl+2}{\dl+2}{\La}\end{matrix} &
		\begin{matrix}\cauchyrowm{l'}{\dl+2}{\dl+2}{\La}\end{matrix}                              \\
		\begin{matrix}\cauchyrowm{k'}{\dl+2}{\dl+3}{\mu}\end{matrix} &
		\begin{matrix}\cauchyrowm{l'}{\dl+2}{\dl+3}{\mu}\end{matrix}                              \\
	\end{bmatrix}.
\]
Therefore, for $\det(B_{k,l}) = 0$,
\begin{align*}
	\det(B_{k,l}) & = \det \begin{bmatrix}
		\begin{matrix}\cauchy{k'}{\dl+2}{\dl+1}\end{matrix} &                            \\
		                           & \begin{matrix}\cauchy{l'}{\dl+2}{\dl+1}\end{matrix} \\
		\begin{matrix}\cauchyrowm{k'}{\dl+2}{\dl+2}{\La}\end{matrix} &
		\begin{matrix}\cauchyrowm{l'}{\dl+2}{\dl+2}{\La}\end{matrix}                              \\
		\begin{matrix}\cauchyrowm{k'}{\dl+2}{\dl+3}{\mu}\end{matrix} &
		\begin{matrix}\cauchyrowm{l'}{\dl+2}{\dl+3}{\mu}\end{matrix}                              \\
	\end{bmatrix} = 0 \\
\end{align*}
\begin{align*}
	 & \Rightarrow \det \begin{bmatrix}
		\det \begin{pmatrix}
			\cauchy{k'}{\dl+2}{\dl+1}          \\
			\cauchyrowm{k'}{\dl+2}{\dl+2}{\La} \\
		\end{pmatrix} &
		\det \begin{pmatrix}
			\cauchy{l'}{\dl+2}{\dl+1}          \\
			\cauchyrowm{l'}{\dl+2}{\dl+2}{\La} \\
		\end{pmatrix}   \\
		\det \begin{pmatrix}
			\cauchy{k'}{\dl+2}{\dl+1}          \\
			\cauchyrowm{k'}{\dl+2}{\dl+3}{\mu} \\
		\end{pmatrix} &
		\det \begin{pmatrix}
			\cauchy{l'}{\dl+2}{\dl+1}          \\
			\cauchyrowm{l'}{\dl+2}{\dl+3}{\mu} \\
		\end{pmatrix}   \\
	\end{bmatrix}=0 \\
	 & \Rightarrow \det \begin{bmatrix}
		\La_{k'} \cdot \frac{\prod_{i \in [\dl+1]}(\B_i-\B_{\dl+2})}{\prod_{i \in k'_S} (\A_i-\B_{\dl+2})} &
		\La_{l'} \cdot \frac{\prod_{i \in [\dl+1]}(\B_i-\B_{\dl+2})}{\prod_{i \in l'_S} (\A_i-\B_{\dl+2})}   \\
		\mu_{k'} \cdot \frac{\prod_{i \in [\dl+1]}(\B_i-\B_{\dl+3})}{\prod_{i \in k'_S} (\A_i-\B_{\dl+3})} &
		\mu_{l'} \cdot \frac{\prod_{i \in [\dl+1]}(\B_i-\B_{\dl+3})}{\prod_{i \in l'_S} (\A_i-\B_{\dl+3})}   \\
	\end{bmatrix}=0  \\
	 & \Rightarrow \det \begin{bmatrix}
		\La_{k'}                                                                &
		\La_{l'}                                                                  \\
		\mu_{k'} \prod_{i \in k'_S} \frac{(\A_i-\B_{\dl+2})}{(\A_i-\B_{\dl+3})} &
		\mu_{l'} \prod_{i \in l'_S} \frac{(\A_i-\B_{\dl+2})}{(\A_i-\B_{\dl+3})}   \\
	\end{bmatrix}=0
\end{align*}
Where $k'_S = \{k'_{1}, \ldots, k'_{\dl+2}\}$ and $l'_S = \{l'_{1}, \ldots, l'_{\dl+2}\}$. The terms $c_{l^{(i)}}$, $d$ and $e_{l^{(i)}}$ were factored out from the above determinant where, \bit
\item $c_{l^{(i)}} = \prod^{}_{f>g, f,g \in  l^{(i)}_S} (\A_f - \A_g)$.
\item $d = \prod_{f>g, f,g \in [\dl+1]}^{} (\B_f - \B_g)$.
\item $e_{l^{(i)}} = \prod_{f \in l^{(i)}_S, g \in [\dl+1]}^{} (\A_f - \B_g)$.
\eit

By the choice of $\A_i$'s, $\prod_{i \in x} \frac{(\A_i-\B_{\dl+2})}{(\A_i-\B_{\dl+3})} \in G$ for $x=k'_S, l'_S$. This yet again is a linear combination of two $\La$'s. Hence this determinant is non-zero and the erasure pattern correctable.

\item Each global erasure is in their own different local code and do not share with the mid-level erasures.

There are four local groups where the erasure occurs, two in each mid-level code. Let the affected mid-level codes be $k$ and $l$ while the local codes within, where the erasure occurs, be $k^{(1)}$ and $k^{(2)}$ and $l^{(1)}$ and $l^{(2)}$ respectively. The matrix $B_{k,l}$ is similar to lemma \ref{lem:2prod-id}

\begin{align*}
	B_{k,l}                                 & = \begin{bmatrix}
		A &   \\
		  & B \\
		C & D \\
	\end{bmatrix}\Rightarrow
	\det(B_{k,l}) = \det \bbm
	\det \bpm A                                                                       \\ C^{(1)}\epm& \det \bpm B \\ D^{(1)}\epm\\
	\det \bpm A                                                                       \\ C^{(2)}\epm& \det \bpm B \\ D^{(2)}\epm\\
	\ebm = 0                                                                          \\
	A                                       & = \bbm
	\begin{matrix}\cauchy{k^{(1)}}{\dl+1}{\dl} \end{matrix}                                                        \\
	                                        & \begin{matrix}\cauchy{k^{(2)}}{\dl+1}{\dl}\end{matrix}              \\
	\begin{matrix}\cauchyrowp{k^{(1)}}{\dl+1}{\dl+1}\end{matrix}              & \begin{matrix}\cauchyrowp{k^{(2)}}{\dl+1}{\dl+1}\end{matrix}              \\
	\ebm                                                                              \\
	B                                       & = \bbm
	\begin{matrix}\cauchy{l^{(1)}}{\dl+1}{\dl} \end{matrix}                                                        \\
	                                        & \begin{matrix}\cauchy{l^{(2)}}{\dl+1}{\dl}\end{matrix}              \\
	\begin{matrix}\cauchyrowp{l^{(1)}}{\dl+1}{\dl+1}\end{matrix}              & \begin{matrix}\cauchyrowp{l^{(2)}}{\dl+1}{\dl+1}\end{matrix}              \\
	\ebm                                                                              \\
	C                                       & = \bbm
	\cauchyrowm{k^{(1)}}{\dl+1}{\dl+2}{\La} & \cauchyrowm{k^{(2)}}{\dl+1}{\dl+2}{\La} \\
	\cauchyrowm{k^{(1)}}{\dl+1}{\dl+3}{\mu} &
	\cauchyrowm{k^{(2)}}{\dl+1}{\dl+3}{\mu}                                           \\
	\ebm                                                                              \\
	D                                       & = \bbm
	\cauchyrowm{l^{(1)}}{\dl+1}{\dl+2}{\La} & \cauchyrowm{l^{(2)}}{\dl+1}{\dl+2}{\La} \\
	\cauchyrowm{l^{(1)}}{\dl+1}{\dl+3}{\mu} &
	\cauchyrowm{l^{(2)}}{\dl+1}{\dl+3}{\mu}                                           \\
	\ebm                                                                              \\
\end{align*}

To calculate the whole determinant, We consider the first element,
\begin{align*}
	\hspace{-8em}
	\det \bpm A                                                                                                                                                                                                                                                                            \\ C^{(1)}\epm &=
	\det \bbm
	\begin{matrix}\cauchy{k^{(1)}}{\dl+1}{\dl} \end{matrix}                                                                                                                                                                                                                                                             \\
	                                                                                                                                   & \begin{matrix}\cauchy{k^{(2)}}{\dl+1}{\dl}\end{matrix}                                                                                                                       \\
	\begin{matrix}\cauchyrowp{k^{(1)}}{\dl+1}{\dl+1}\end{matrix}                                                                                                        & \begin{matrix}\cauchyrowp{k^{(2)}}{\dl+1}{\dl+1}\end{matrix}                                                                                                                       \\
	\begin{matrix}\cauchyrowm{k^{(1)}}{\dl+1}{\dl+2}{\La}\end{matrix}                                                                                                        & \begin{matrix}\cauchyrowm{k^{(2)}}{\dl+1}{\dl+2}{\La}\end{matrix}                                                                                                                       \\
	\ebm                                                                                                                                                                                                                                                                                   \\
	                                                                                                                                   & = 	\scalemath{0.95}{\det \bbm
	\det \bpm \cauchy{k^{(1)}}{\dl+1}{\dl}                                                                                                                                                                                                                                                 \\ \cauchyrowp{k^{(1)}}{\dl+1}{\dl+1} \epm &
	\det \bpm \cauchy{k^{(2)}}{\dl+1}{\dl}                                                                                                                                                                                                                                                 \\ \cauchyrowp{k^{(2)}}{\dl+1}{\dl+1} \epm \\
	\det \bpm \cauchy{k^{(1)}}{\dl+1}{\dl}                                                                                                                                                                                                                                                 \\
	\cauchyrowm{k^{(1)}}{\dl+1}{\dl+2}{\La} \epm                                                                                       &
	\det \bpm \cauchy{k^{(2)}}{\dl+1}{\dl}                                                                                                                                                                                                                                                 \\\cauchyrowm{k^{(2)}}{\dl+1}{\dl+2}{\La} \epm \\
	\ebm}                                                                                                                                                                                                                                                                                  \\
	                                                                                                                                   & =	\det \bbm
	\frac{c_{k^{(1)}}d\prod_{i \in [\dl]}(\B_i-\B_{\dl+1})}{e_{k^{(1)}} \prod_{i \in k^{(1)}_S} (\A_i-\B_{\dl+1})}                     &
	\frac{c_{k^{(2)}}d\prod_{i \in [\dl]}(\B_i-\B_{\dl+1})}{e_{k^{(2)}} \prod_{i \in k^{(2)}_S} (\A_i-\B_{\dl+1})}                                                                                                                                                                         \\
	\La_{k^{(1)}} \cdot \frac{c_{k^{(1)}}d\prod_{i \in [\dl]}(\B_i-\B_{\dl+2})}{e_{k^{(1)}} \prod_{i \in k^{(1)}_S} (\A_i-\B_{\dl+2})} &
	\La_{k^{(2)}} \cdot \frac{c_{k^{(2)}}d\prod_{i \in [\dl]}(\B_i-\B_{\dl+2})}{e_{k^{(2)}} \prod_{i \in k^{(2)}_S} (\A_i-\B_{\dl+2})}                                                                                                                                                     \\
	\ebm                                                                                                                                                                                                                                                                                   \\
	                                                                                                                                   & = \scalemath{1}{\frac{c_{k^{(1)}}c_{k^{(2)}}d^2}{e_{k^{(1)}} e_{k^{(2)}}} \prod_{i \in [\dl]}(\B_i - \B_{\dl+1})(\B_i-\B_{\dl+2}) \cdot \det \bbm
	\prod_{i \in k^{(1)}_S} \frac{1}{\A_i-\B_{\dl+1}}                                                                                  & \prod_{i \in k^{(2)}_S} \frac{1}{\A_i-\B_{\dl+1}}                                                                                                 \\
	\La_{k^{(1)}} \cdot \prod_{i \in k^{(1)}_S} \frac{1}{\A_i-\B_{\dl+2}}                                                              & \La_{k^{(2)}} \cdot \prod_{i \in k^{(2)}_S} \frac{1}{\A_i-\B_{\dl+2}}                                                                             \\
	\ebm}.                                                                                                                                                                                                                                                                                 \\
\end{align*}
Where, \bit
\item $k^{(i)}_S = \{k^{(i)}_{1}, \ldots, k^{(i)}_{\dl+2}\}$.
\item $c_{k^{(i)}} = \prod^{}_{f>g, f,g \in  k^{(i)}_S} (\A_f - \A_g)$.
\item $d = \prod_{f>g, f,g \in [\dl]}^{} (\B_f - \B_g)$.
\item $e_{k^{(i)}} = \prod_{f \in k^{(i)}_S, g \in [\dl]}^{} (\A_f - \B_g)$.
\eit

Applying all this in the main determinant and setting,
\[
	\det(B_{k,l}) = 0
\]
and factoring out the common multiples, we get
\begin{align*}
	\hspace{-8em}
	\scalemath{0.80}{\det \bbm
		\det\bpm
	\prod_{i \in k^{(1)}_S} \frac{1}{\A_i-\B_{\dl+1}}                                                                                                                         & \prod_{i \in k^{(2)}_S} \frac{1}{\A_i-\B_{\dl+1}}                     \\
	\La_{k^{(1)}} \cdot \prod_{i \in k^{(1)}_S} \frac{1}{\A_i-\B_{\dl+2}}                                                                                                     & \La_{k^{(2)}} \cdot \prod_{i \in k^{(2)}_S} \frac{1}{\A_i-\B_{\dl+2}} \\
	\epm                                                                                                                                                                      &
		\det\bpm
	\prod_{i \in l^{(1)}_S} \frac{1}{\A_i-\B_{\dl+1}}                                                                                                                         & \prod_{i \in l^{(2)}_S} \frac{1}{\A_i-\B_{\dl+1}}                     \\
	\La_{l^{(1)}} \cdot \prod_{i \in l^{(1)}_S} \frac{1}{\A_i-\B_{\dl+2}}                                                                                                     & \La_{l^{(2)}} \cdot \prod_{i \in l^{(2)}_S} \frac{1}{\A_i-\B_{\dl+2}} \\
	\epm                                                                                                                                                                                                                                              \\
		\det\bpm
	\prod_{i \in k^{(1)}_S} \frac{1}{\A_i-\B_{\dl+1}}                                                                                                                         & \prod_{i \in k^{(2)}_S} \frac{1}{\A_i-\B_{\dl+1}}                     \\
	\mu_{k^{(1)}} \cdot \prod_{i \in k^{(1)}_S} \frac{1}{\A_i-\B_{\dl+3}}                                                                                                     & \mu_{k^{(2)}} \cdot \prod_{i \in k^{(2)}_S} \frac{1}{\A_i-\B_{\dl+3}} \\
	\epm                                                                                                                                                                      &
		\det\bpm
	\prod_{i \in l^{(1)}_S} \frac{1}{\A_i-\B_{\dl+1}}                                                                                                                         & \prod_{i \in l^{(2)}_S} \frac{1}{\A_i-\B_{\dl+1}}                     \\
	\mu_{l^{(1)}} \cdot \prod_{i \in l^{(1)}_S} \frac{1}{\A_i-\B_{\dl+3}}                                                                                                     & \mu_{l^{(2)}} \cdot \prod_{i \in l^{(2)}_S} \frac{1}{\A_i-\B_{\dl+3}} \\
	\epm                                                                                                                                                                                                                                              \\
	\ebm=0}                                                                                                                                                                                                                                           \\
	\det \scalemath{0.9}{\bbm
	\La_{k^{(2)}} \cdot \prod_{i \in k^{(2)}_S} \frac{\A_i-\B_{\dl+1}}{\A_i-\B_{\dl+2}} - \La_{k^{(1)}} \cdot \prod_{i \in k^{(1)}_S} \frac{\A_i-\B_{\dl+1}}{\A_i-\B_{\dl+2}} &
	\La_{l^{(2)}} \cdot \prod_{i \in l^{(2)}_S} \frac{\A_i-\B_{\dl+1}}{\A_i-\B_{\dl+2}} - \La_{l^{(1)}} \cdot \prod_{i \in l^{(1)}_S} \frac{\A_i-\B_{\dl+1}}{\A_i-\B_{\dl+2}}                                                                         \\
	\mu_{k^{(2)}} \cdot \prod_{i \in k^{(2)}_S} \frac{\A_i-\B_{\dl+1}}{\A_i-\B_{\dl+3}} - \mu_{k^{(1)}} \cdot \prod_{i \in k^{(1)}_S} \frac{\A_i-\B_{\dl+1}}{\A_i-\B_{\dl+3}} &
	\mu_{l^{(2)}} \cdot \prod_{i \in l^{(2)}_S} \frac{\A_i-\B_{\dl+1}}{\A_i-\B_{\dl+3}} - \mu_{l^{(1)}} \cdot \prod_{i \in l^{(1)}_S} \frac{\A_i-\B_{\dl+1}}{\A_i-\B_{\dl+3}}                                                                         \\
		\ebm=0}
\end{align*}
Where similarly, $l^{(i)}_S = \{l^{(i)}_{1}, \ldots, l^{(i)}_{\dl+2}\}$.

Now, since the $\La_i$'s are 4-wise independent over $\Fqo$, the first row is never zero. Similarly, all the $\mu_j$'s are in different cosets of $G$ and by choice of $\A$'s
$\prod_{i \in l^{(j)}_S, k^{(j)}_S} \frac{\A_i-\B_{\dl+1}}{\A_i-\B_{\dl+3}} \in G$. Hence the last row isn't zero either. Then this determinant resolves into a linear combination for 4 different values of $\La_i$s. Hence, by linear independence rules of $\La$, this determinant is also non-zero.

\item In this case, one of the global erasure shares the local code with a mid-level code while the other does not. Assume that the $k^{th}$ and $l^{th}$ mid-level codes are affected.
Let the local codes within, where the erasure occurs, be $k^{(1)}, k^{(2)}$ and $l^{(1)}$. The matrix $B_{k,l}$,

\[\hspace{-6em}
	\scalemath{0.9}{
		B_{k,l} = \begin{bmatrix}
			\begin{matrix}\cauchy{k^{(1)}}{\dl+1}{\dl} \end{matrix}                                                             \\

			                            & \begin{matrix}\cauchy{k^{(2)}}{\dl+1}{\dl}\end{matrix}                               \\

			\begin{matrix}\cauchyrowp{k^{(1)}}{\dl+1}{\dl+1}\end{matrix} & \begin{matrix}\cauchyrowp{k^{(2)}}{\dl+1}{\dl+1}\end{matrix}                               \\

			                            &                             & \begin{matrix}\cauchy{l^{(1)}}{\dl+2}{\dl+1}\end{matrix} \\

			\begin{matrix}\cauchyrowm{k^{(1)}}{\dl+1}{\dl+2}{\La}\end{matrix} & \begin{matrix}\cauchyrowm{k^{(2)}}{\dl+1}{\dl+2}{\La}\end{matrix} & \begin{matrix}\cauchyrowm{l^{(1)}}{\dl+2}{\dl+2}{\La}\end{matrix} \\

			\begin{matrix}\cauchyrowm{k^{(1)}}{\dl+1}{\dl+3}{\mu}\end{matrix} & \begin{matrix}\cauchyrowm{k^{(2)}}{\dl+1}{\dl+3}{\mu}\end{matrix} & \begin{matrix}\cauchyrowm{l^{(1)}}{\dl+2}{\dl+3}{\mu}\end{matrix} \\
		\end{bmatrix}
	}\]

Now, after permuting one row, we can apply \ref{lem:3prod-id} to expand the matrix for the determinant,
\begin{align*}
	\hspace{-8em}
	\det(B_{k,l}) = &                                                   \\
	                & \scalemath{0.75}{\det \begin{pmatrix}
			\cauchy{k^{(1)}}{\dl+1}{\dl}       \\
			\cauchyrowp{k^{(1)}}{\dl+1}{\dl+1} \\
		\end{pmatrix}
		\det \begin{pmatrix}
			\cauchy{k^{(2)}}{\dl+1}{\dl} \\
			\cauchyrowm{k^{(2)}}{\dl+1}{\dl+2}{\La}
		\end{pmatrix}
	\det \begin{pmatrix}
			\cauchy{l^{(1)}}{\dl+2}{\dl}            \\
			\cauchyrowp{l^{(1)}}{\dl+2}{\dl+1}      \\
			\cauchyrowm{l^{(1)}}{\dl+2}{\dl+3}{\mu} \\
		\end{pmatrix}} +                                 \\
	                & \scalemath{0.75}{\det \begin{pmatrix}
			\cauchy{k^{(1)}}{\dl+1}{\dl} \\
			\cauchyrowp{k^{(1)}}{\dl+1}{\dl+1}
		\end{pmatrix}
		\det \begin{pmatrix}
			\cauchy{k^{(2)}}{\dl+1}{\dl} \\
			\cauchyrowm{k^{(2)}}{\dl+1}{\dl+3}{\mu}
		\end{pmatrix}
	\det \begin{pmatrix}
			\cauchy{l^{(1)}}{\dl+2}{\dl}       \\
			\cauchyrowp{l^{(1)}}{\dl+2}{\dl+1} \\
			\cauchyrowm{l^{(1)}}{\dl+2}{\dl+2}{\La}
		\end{pmatrix}} +                                 \\
	                & \scalemath{0.75}{\det \begin{pmatrix}
			\cauchy{k^{(1)}}{\dl+1}{\dl}            \\
			\cauchyrowm{k^{(1)}}{\dl+1}{\dl+2}{\La} \\
		\end{pmatrix}
		\det \begin{pmatrix}
			\cauchy{k^{(1)}}{\dl+1}{\dl}       \\
			\cauchyrowp{k^{(2)}}{\dl+1}{\dl+1} \\
		\end{pmatrix}
	\det \begin{pmatrix}
			\cauchy{l^{(1)}}{\dl+2}{\dl}            \\
			\cauchyrowp{l^{(1)}}{\dl+2}{\dl+1}      \\
			\cauchyrowm{l^{(1)}}{\dl+2}{\dl+3}{\mu} \\
		\end{pmatrix}} -                                 \\
	                & \scalemath{0.75}{\det \begin{pmatrix}
			\cauchy{k^{(1)}}{\dl+1}{\dl}            \\
			\cauchyrowm{k^{(1)}}{\dl+1}{\dl+3}{\mu} \\
		\end{pmatrix}
		\det \begin{pmatrix}
			\cauchy{k^{(2)}}{\dl+1}{\dl}       \\
			\cauchyrowp{k^{(2)}}{\dl+1}{\dl+1} \\
		\end{pmatrix}
	\det \begin{pmatrix}
			\cauchy{l^{(1)}}{\dl+2}{\dl}            \\
			\cauchyrowp{l^{(1)}}{\dl+2}{\dl+1}      \\
			\cauchyrowm{l^{(1)}}{\dl+2}{\dl+2}{\La} \\
		\end{pmatrix}}                                   \\
\end{align*}
Now, in this massive expansion, we can take $\La_{i}$ and $\mu_{j}$ out of the determinants. What we will find is that each term is $\La_{i}\mu_{j}$ multiplied by the product of the determinant of three Cauchy matrices. Each of those determinant $\in \Fqo$.

Hence the final determinant is actually the linear combination of three $\La_i$ in $\Fqo$. Hence $\det(B_{k,l}) \neq 0$.
\een
\een

\medskip

\bibliography{biblo}{}
\bibliographystyle{ieeetr}

\end{document}